\documentclass{birkjour}
\usepackage[square,sort,comma,numbers]{natbib}

\usepackage{latexsym}
\usepackage{amsmath}
\usepackage{epsfig}
\usepackage{amssymb}
\usepackage{enumerate}
\usepackage{pst-all}
\usepackage{multido}
\usepackage{pst-blur}
\usepackage{color}

\usepackage{color}
\usepackage{graphicx}
\usepackage{amsmath}
\usepackage{amssymb}

\usepackage{latexsym}

\usepackage{amsfonts,amsgen,amsopn}
\usepackage{epsfig}
\usepackage{epstopdf}
\usepackage{mathpazo}

\DeclareMathOperator*{\argmax}{arg\,max}

\newtheorem{theorem}{Theorem}[section]
\newtheorem{lemma}[theorem]{Lemma}

\newtheorem{corollary}[theorem]{Corollary}

\newtheorem{definition}[theorem]{Definition}

\newtheorem{observation}[theorem]{Observation}

\newcommand{\T}{{\mathcal T}}
\newcommand{\C}{{\mathcal C}}
\newcommand{\NN}{{\mathbb N}}

 \theoremstyle{definition}
 
 \theoremstyle{remark}

 \numberwithin{equation}{section}

\begin{document}

%
%
%
%
%
%
%
%
%

\title[The MP distance between phylogenetic trees]{On the Maximum Parsimony distance\\between phylogenetic trees}


\author{Mareike Fischer}

\address{Ernst-Moritz-Arndt University of Greifswald \\ Department for Mathematics and Computer Science \\ Walther-Rathenau-Str. 47 \\ 17487 Greifswald, Germany}
\email{email@mareikefischer.de}

\author{Steven Kelk}
\address{Department of Knowledge Engineering (DKE)\\ Maastricht University\\ P.O. Box 616, 6200
MD Maastricht\\ The Netherlands}
\email{steven.kelk@maastrichtuniversity.nl}


\subjclass{05C15; 05C35; 90C35; 92D15}

\keywords{Maximum Parsimony, Tree Metric, Subtree prune and regraft (SPR)}

\date{\today}


\begin{abstract}
Within the field of phylogenetics there is great interest in distance measures to quantify the dissimilarity of two trees. Here, based on an idea of Bruen and Bryant,
we propose and analyze a new distance measure: the Maximum Parsimony (MP) distance.
This is based on the difference of the parsimony scores of a single character on both trees under consideration, and the goal is to find the character which maximizes this difference. In this article we show that this new distance is a metric and provides a lower bound to the well-known Subtree Prune and Regraft (SPR) distance. We also
show that to compute the MP distance it is sufficient to consider only characters that are convex on one of the trees,
and prove several additional structural properties of the distance. On the complexity side, we prove that calculating the MP distance is in general NP-hard, and identify an interesting island of tractability in which the distance can be calculated in polynomial time.
\end{abstract}
\maketitle

\section{Introduction}\label{sec:intro}

Finding the optimal tree explaining a given dataset, e.g. a DNA alignment, is one of the biggest challenges in modern phylogenetics. One challenge is that for most optimization criteria finding the best tree is NP-hard (cf. \citep{foulds_graham_1982,roch_2006,chor_tuller_2006}). For this reason heuristics based
on local neighborhood search are often used. A second challenge concerns the fact that, for mathematical and/or biological reasons, many distinct tree solutions may be generated, and some way of determining their relative similarity is required. Both problems can be addressed using SPR (Subtree-Prune-and-Regraft). A single SPR move involves moving to a neighboring tree by detaching a branch and re-attaching it elsewhere. The SPR distance is the minimum number of SPR moves required to transform one tree into another.
Unfortunately, computing the SPR distance is hard {\citep{bordewich_semple_2004}, \citep{bonet_stjohn_2010}}. Moreover, for SPR it makes a big difference whether rooted or unrooted trees are considered \citep{bonet_stjohn_2010}. Other metrics, which can be calculated in polynomial time, have been proposed, like e.g. the so-called Robinson-Foulds metric \citep{robinson_foulds_1981}, but are also sometimes criticized for lack of biological plausibility {\citep{lin2012}}. 

In this paper, we propose a new metric, namely the Maximum Parsimony (MP) distance, which is biologically feasible in the parsimony sense. Our metric basically requires the search for a character which has a low parsimony score on one of the trees involved and a high score on the other one. As the parsimony score is independent of the root position, our metric applies both to rooted and unrooted trees. We analyze a number of structural properties of the metric, and explore the computational complexity of computing it. Indeed, this research was initially inspired by a question posed by Bruen and Bryant: is the MP distance efficiently computable, and could it act as a tractable approximation of the SPR distance?
We show that there is a link, at least in one direction, between the two measures: the MP distance provides a lower bound on the SPR metric.  Moreover, we show that finding the character maximizing the difference in parsimony performance on two trees can be achieved by considering only characters which are convex on one of the trees. Relatedly, we provide combinatorial bounds on the MP distance and observe that to compute this distance it is not sufficient to restrict our search to characters with a fixed number of states.

Despite the possibility of restricting the problem to convex characters, we prove that calculating the MP distance is, unfortunately, NP-hard. This hardness also holds when characters with at most two states are considered; interestingly neither hardness result is directly implied by the other. On the positive side we show, by exploiting a classical result from the tree partitioning literature, that the MP distance can be computed in polynomial time when one of the trees is a so-called star tree.

\section{Preliminaries and Notation}\label{sec:prelim}
We need to introduce some notation before presenting our results.

Recall that an {\it unrooted phylogenetic $X$-tree} is a tree $\T =(V(\T),E(\T))$ on a leaf set $X=\{1,\ldots,m\} \subset V(\T)$. Such a tree is named {\it binary} if it has only vertices of degree 1 (leaves) or 3 (internal vertices). A  {\it rooted phylogenetic $X$-tree} additionally has one vertex specified as the {\it root}, and such a rooted tree is named {\it binary} if the root has degree 2 and all other vertices are of degree 1 (leaves) or 3 (internal vertices). We often denote trees in the well-known Newick format \citep{felsenstein_2000}, which uses nested parentheses to group species together according to their degree of relatedness. For instance, the tree $((1,2),(3,4))$ is a tree with two so-called cherries $(1,2)$ and $(3,4)$ and a root between the two. Unrooted trees have more than two groups of parentheses at the uppermost level.

Furthermore, recall that a {\it character} $f$ is a function $f: X\rightarrow \C$ for some set $\C:=\{c_1, c_2, c_3, \ldots, c_k \}$ of $k$ {\em character states} ($k \in \NN$). Often, $k$ is assumed to equal 4 in order for $\C$ to represent the DNA alphabet $\{A,C,G,T\}$, but in the present paper $k$ is not restricted this way but can be any natural number. Note that in the special cases where $|f(X)|=2$, $|f(X)|=3$ or $|f(X)|=4$, we also refer to $f$ as a binary, ternary or quaternary character, respectively. In general, when $|f(X)|=r$, $f$ is called an \emph{$r$-state character}. In order to shorten the notation, it is customary to write for instance $f=AACC$ instead of $f(1)=A$, $f(2)=A$, $f(3)=C$ and $f(4)=C$. Note that each $r$-state character $f$ on taxon set $X$ partitions $X$ into $r$ non-empty and non-overlapping subsets $X_i$, $i=1,\ldots,r$, where $x_j,x_k \in X_i$ if and only if $f(x_j)=f(x_k)$. When $r=2$, the resulting bipartition of $X$ into the non-empty and disjoint sets $X_1$ and $X_2$ is called {\it $X$-split} and denoted by $X_1|X_2$. Note that the branches of a tree $\T$ induce a collection of $X$-splits: each edge separates some leaves from the others. Thus, labelling the leaves on one side of the branch with, say, $A$, and the others, say, with $C$, gives a binary character. The collection of all splits induced by $\T$ this way will be referred to as $\Sigma(\T)$. A tree $\T_1$ is called a {\it refinement} of another tree $\T_2$, if $\Sigma(\T_2)\subseteq \Sigma(\T_1)$. The {\em star tree} is a tree $\T$ which has only one internal node with which all leaves are directly connected via a pendant edge, i.e. all splits induced by $\T$ are of the kind $x| X \setminus \{x\}$. Note that all phylogenetic $X$-trees are refinements of the star tree.

A {\it refinement} of an $r$-state character $f$ on $X$ is an $\hat{r}$-state character $\hat{f}$ on $X$ with $r\leq\hat{r}$ such that the partitioning induced by $\hat{f}$ refines that given by $f$. This means that if $f$ induces the partitioning $X_1|X_2|\ldots|X_r$ and $\hat{f}$ induces the partitioning $Y_1|Y_2|\ldots|Y_{\hat{r}}$, where $X_i,Y_j$ are subsets of $X$ for all $i=1,\ldots,r$, $j=1,\ldots,\hat{r}$, then for all $j=1,\ldots,\hat{r}$ there exists an $i\in\{1,\ldots,r\}$ such that $Y_j \subseteq X_i$.

An {\it extension} of $f$ to $V(\T)$ is a map $g: V(\T)\rightarrow \C$ such that $g(i) = f(i)$ for all $i$ in $X$. For such an extension $g$ of $f$, we denote by $l_{g}(\T)$ the number of edges $e=\{u,v\}$ in $\T$ on which a {\em substitution} occurs, i.e. where $g(u) \neq g(v)$. Such substitutions are also often referred to as {\em mutations} or {\em changes}. 
The {\em parsimony score} or {\em parsimony length} of a character $f$ on $\T$, denoted by $l_{f}(\T)$, is obtained by minimizing $l_{g}(\T)$ over all possible extensions $g$ of $f$. The parsimony score of a character $f$ on a phylogenetic tree $\T$ can easily be calculated with the Fitch algorithm \citep{fitch_1971} if $\T$ is binary. Moreover, the Fitch algorithm was generalized by Hartigan \citep{hartigan_1973} to apply also to non-binary trees. In order to simplify the notation, we will refer to both algorithms as the Fitch algorithm rather than the Fitch-Hartigan or generalized Fitch algorithm. (The only significant difference between the two algorithms is that in the non-binary
algorithm, during the bottom-up phase, a parent is allocated all states that occur \emph{most frequently} amongst its children. Letting $m$ denote the number of times a most frequent state occurs amongst the children, the number of mutations incurred is equal to the
number of children minus $m$. This generalises the intersection/union operations used by
the binary algorithm.)

Note that the Fitch algorithm can be applied both to rooted an unrooted
trees -- in the latter case, the tree can be rooted by placing an extra root node on an arbitrary edge of the tree. This implies that the parsimony score does not depend on the root position and that for the parsimony concept it does not matter if we discuss rooted or unrooted trees. This is the reason why the MP distance, which we define shortly, is unaffected by the presence, or location, of a root.

A character $f$ is said to be {\it convex} or {\it homoplasy-free} on a tree $\T$ if $l_{f}(\T)=|f|-1=r-1$, where $|f|=r$ denotes the number of character states employed by $f$. Note that if a character is convex on a certain tree, this tree minimizes its parsimony score and is therefore most parsimonious for this character, respectively. Moreover, recall that two characters are {\it compatible} if there exists a phylogenetic $X$-tree on which both of them are convex, and two splits are said to be compatible if there is a phylogenetic $X$-tree which contains both branches corresponding to the splits.

Recall that a character $f$ on a leaf set $X$ is said to be {\it informative} (with respect to parsimony) if at least two distinct character states occur more than once on $X$. Otherwise $f$ is called {\it non-informative}. Note that for a non-informative character $f$, $l_{f}(\T_i)=l_{f}(\T_j)$ for
all trees $\T_i$, $\T_j$ on the same set $X$ of leaves. 

In this paper, we refer to a character always with its underlying taxon clustering pattern in mind, i.e. for instance we do not distinguish between $AACC$, $CCAA$ and $CCGG$, and so on. Moreover, when there is no ambiguity and when the stated result holds for both rooted and unrooted trees, we often just write `tree' or `phylogenetic tree' when referring to a phylogenetic $X$-tree.

Recall that a {\it subtree prune and regraft (SPR) move} on a phylogenetic tree $\T$ is defined for unrooted trees according to, e.g. \citep{allen_steel_2001,bruen_bryant_2008}, and for rooted trees according to, e.g. \citep{linz_semple_2011}, as cutting any edge and thereby pruning a subtree, $\tilde{\T}$, and then regrafting the subtree by the same cut edge to a new vertex obtained by subdividing a pre-existing edge in $\T\setminus \tilde{\T}$. If $\T$ is binary, one can suppress degree 2 vertices in order for the resulting tree to be binary, too. We define the {\it SPR distance} $d_{SPR}$ of two unrooted phylogenetic trees $\T_1$, $\T_2$ as in \citep{bruen_bryant_2008} as the minimum number of SPR moves needed to change $\T_1$ into $\T_2$. Note that for SPR, it does make a difference whether the trees under consideration are rooted or not. When two trees have a different root position but are otherwise identical, $d_{SPR}$ is 0, but their rooted SPR distance is greater than 0. If we refer to the rooted SPR distance as defined in \citep{linz_semple_2011}, we therefore explicitly write $d_{rSPR}$.
When discussing the relationship between MP distance and SPR we shall
restrict our analysis to binary trees. This is because $d_{rSPR}$ on non-binary trees is a relatively unknown measure (although see \citep{iersel2014,collins2009,Rodrigues2007}) and is usually defined such that refinements of the original trees are permitted. These refinements lead to major technicalities and can be shown to severely weaken the relationship between MP distance and $d_{rSPR}$. Moreover, in the unrooted context, there are no major results available on $d_{SPR}$ on non-binary trees. Thus, the present paper focusses on binary SPR.

\section{Structural properties of MP distance}\label{sec:results}

\subsection{The Maximum Parsimony distance between phylogenetic trees}\label{sec:MPdist}
\subsubsection{Definition and basic properties}\label{sec:def}

We are now in a position to introduce the concept of measuring the distance between phylogenetic trees as follows.

\begin{definition}\label{def:MPdist} Let $\T_1$, $\T_2$ be two (rooted or unrooted) phylogenetic trees on a set $X$ of taxa with $|X|=n$. Then, $$d_{MP}(\T_1,\T_2):=\max\limits_{f}|l_f(\T_1)-l_f(\T_2)|$$

describes the Maximum Parsimony distance or MP distance between $\T_1$ and $\T_2$, where the maximum is taken over all characters $f$ on taxon set $X$.
\end{definition}

Where it is unambiguous from the context we will say that a character $f$ is \emph{optimal} for $\T_1, \T_2$ if $|l_f(\T_1)-l_f(\T_2)| = d_{MP}(\T_1, \T_2)$.

Before we continue with some properties of $d_{MP}$, we note that the absolute value in the definition of $d_{MP}$ is necessary in order to achieve symmetry: Consider $\T_1=(((1,3),(2,5)),(((4,7),6),8))$ and $\T_2=((((2,6),(3,5)),(1,4)),$\\$(7,8))$. These two trees are depicted by Figure \ref{antisym:ex}. In this case, the character $f=ACAGCGGA$ gives a difference of $l_f(\T_2)-l_f(\T_1)=5-2=3$. An exhaustive search through all characters on 8 taxa reveals that this is maximum; however, it also shows that there is no character $\tilde{f}$ such that $l_{\tilde{f}}(\T_1)-l_{\tilde{f}}(\T_2)=3$, as all characters on 8 taxa give at most a difference of 2. Thus, the roles of $\T_1$ and $\T_2$ cannot simply be swapped, which is why the absolute value in the definition of $d_{MP}$ is required.

\begin{figure}[ht]      \centering\vspace{0.5cm} 
       \includegraphics[width=10cm]{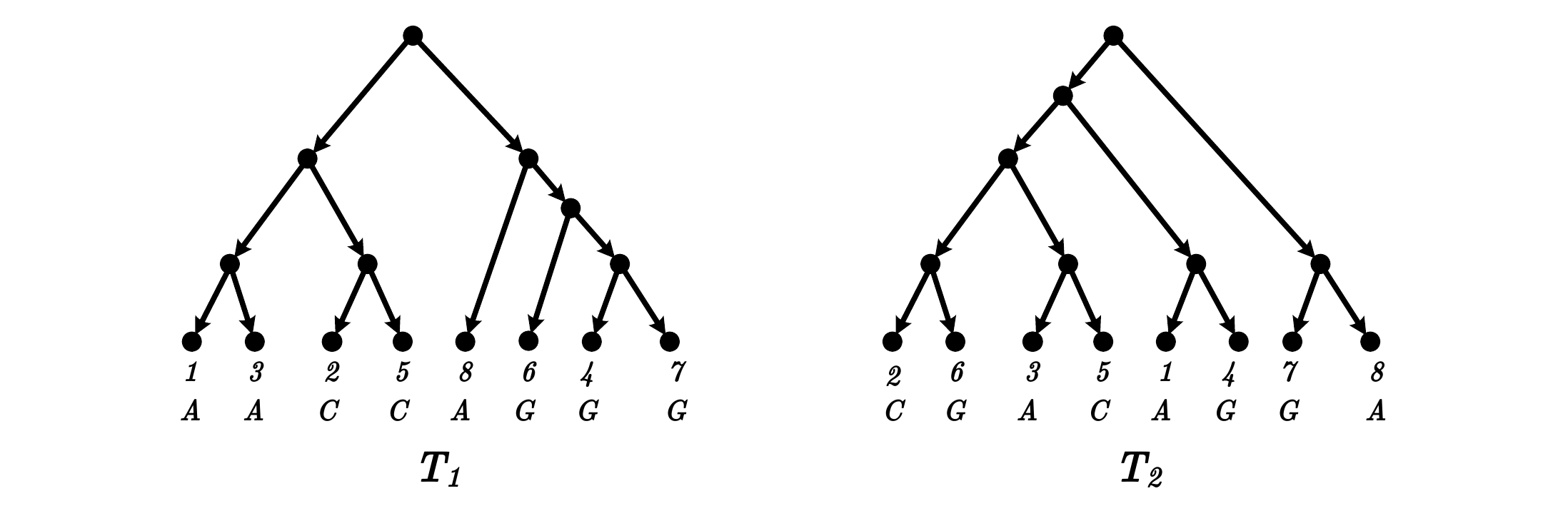}
   \caption{Two rooted binary phylogenetic $X$-trees on the same set of eight taxa, for which the character $f=ACAGCGGA$ gives a difference of $l_f(\T_2)-l_f(\T_1)=5-2=3$, but for which there is no character $\tilde{f}$ such that $l_{\tilde{f}}(\T_1)-l_{\tilde{f}}(\T_2)=3$.}\label{antisym:ex}
\end{figure}

We now formally state that the definition of $d_{MP}$ is indeed a metric.

\begin{theorem}\label{thm:metric}
The MP distance as defined in Definition \ref{def:MPdist} is a metric, i.e. it fulfills the following properties:
\begin{enumerate}
\item \label{prop:nonneg} $d_{MP}(\T_1,\T_2)\geq 0$ for all phylogenetic trees $\T_1,\T_2$,
\item \label{prop:equality} $d_{MP}(\T_1,\T_2) = 0$ if and only if $\T_1=\T_2$,
\item \label{prop:symmetry} $d_{MP}(\T_1,\T_2)=d_{MP}(\T_2,\T_1)$ for all phylogenetic trees $\T_1,\T_2$,
\item \label{prop:triangleineq} $d_{MP}(\T_1,\T_3)\leq d_{MP}(\T_1,\T_2)+d_{MP}(\T_2,\T_3)$ for all phylogenetic trees $\T_1,\T_2,\T_3$.
\end{enumerate}

\end{theorem}

Before we can prove this theorem, we recall the following theorem from \citep[Theorem 3.1.4]{semple_steel_2003}, which goes back to \citep{buneman_1971}:

\begin{theorem}[Splits Equivalence Theorem]\label{thm:splitseq} Let $\Sigma$ be a collection of $X$-splits. Then, there is an $X$-tree $\T$ such that $\Sigma(\T)=\Sigma$ if and only if the splits in $\Sigma$ are pairwise compatible. Moreover, if such an $X$-tree exists, then, up to isomorphism, $\T$ is unique.
\end{theorem}

We are now in a position to prove Theorem \ref{thm:metric}.

\begin{proof}
Properties \ref{prop:nonneg} and \ref{prop:symmetry} are clear by the usage of the absolute value in the definition of $d_{MP}$. 

Now consider Property \ref{prop:equality}. Let $\T_1$, $\T_2$ be two phylogenetic $X$-trees with $|X|=n$. If $\T_1=\T_2$, then for all characters $f$ on $X$ we have $l_f(\T_1)=l_f(\T_2)$ and therefore $d_{MP}(\T_1,\T_2)=0$. This completes the first direction. If, on the other hand, $d_{MP}(\T_1,\T_2)=0$, this implies by the definition of $d_{MP}$ and by Property \ref{prop:nonneg} that $l_f(\T_1)=l_f(\T_2)$ for all characters $f$ on $X$. In particular, all splits induced by $\T_1$ and their corresponding characters give the same parsimony score, namely 1, on $\T_2$. Thus, all binary characters induced by the branches of $\T_1$ are convex on $\T_2$ and thus are compatible with the binary characters induced by the splits of $\T_2$ and vice versa. So the entire collection of splits induced by both $\T_1$ and $\T_2$ is convex on both $\T_1$ and $\T_2$ and the splits are therefore in particular pairwise compatible. However, by Theorem \ref{thm:splitseq}, this implies $\T_1=\T_2$. 

Next we prove the triangle inequality stated by Property \ref{prop:triangleineq}. Let $\T_1,\T_2,\T_3$ be phylogenetic $X$-trees. Let $\hat{f}:=\argmax\limits_{f}|l_f(\T_1)-l_f(\T_3)|$, i.e. $\hat{f}$ is a character which gives $d_{MP}(\T_1,\T_3)$. Without loss of generality, assume $l_{\hat{f}}(\T_1)\geq l_{\hat{f}}(\T_3)$, which implies $d_{MP}(\T_1,\T_3)=l_{\hat{f}}(\T_1)-l_{\hat{f}}(\T_3)$. Now consider $\T_2$:

$$d_{MP}(\T_1,\T_2)+d_{MP}(\T_2,\T_3) = \max\limits_{f} |l_f(\T_1)-l_f(\T_2)|  + \max\limits_{f} |l_f(\T_2)-l_f(\T_3)|   $$
$$ \geq (l_{\hat{f}}(\T_1)-l_{\hat{f}}(\T_2))+ (l_{\hat{f}}(\T_2)-l_{\hat{f}}(\T_3))=l_{\hat{f}}(\T_1)-l_{\hat{f}}(\T_3)=d_{MP}(\T_1,\T_3).$$

Here, the inequality is due to the definition of $d_{MP}$ with the absolute value and the maximum, i.e. the inequality holds for any character on $X$ and thus in particular for $\hat{f}$. This completes the proof.
\end{proof}

\subsubsection{The relationship of the MP distance and the SPR distance}\label{sec:spr}

The idea of introducing a new distance measure on the tree space was motivated by the search for bounds on the SPR distance. The fact that our MP distance indeed provides a lower bound for the (unrooted) SPR distance can be concluded by building on \citep[Theorem 1]{bruen_bryant_2008}: 

\begin{theorem}[Bruen and Bryant, Theorem 1]\label{thm:bruenbryant} Let $f$ be a character with $r$ states on a taxon set $X$ and let $\T$ be an unrooted binary phylogenetic $X$-tree. It takes exactly $l_f(\T)-(r-1)$ SPR moves to transform $\T$ into a tree on which $f$ is convex.
\end{theorem}

Before we can use Theorem \ref{thm:bruenbryant} to prove that our new metric is a lower bound for the SPR distance, we first introduce a general observation on refinements of characters, which subsequently helps us to simplify the search for the character that maximizes the MP distance between any two trees.

\begin{lemma} \label{lem:refine} Let $\T$ be a phylogenetic $X$-tree, let $f$ be a character on $\T$ and $\tilde{f}$ a refinement of $f$. Then, $l_f(\T) \leq l_{\tilde{f}}(\T)$.\end{lemma}

\begin{proof} If $f=\tilde{f}$, there is nothing to show. Now let the partitioning induced by $f$ be $X_1|X_2|\ldots|X_r$ and the partitioning induced by $\tilde{f}$ be $Y_1|Y_2|\ldots|Y_{\tilde{r}}$ with $\tilde{r} > r$, i.e. $\tilde{f}$ is a strict refinement of $f$. By definition of a refinement, every $Y_j$ is contained in an $X_i$ for some $i$. Now as $\tilde{f}$ strictly refines $f$, we can assume without loss of generality that $X_1$ contains $Y_1$ and $Y_2$. Now let $\tilde{g_f}$ be a most parsimonious extension of $\tilde{f}$ on $\T$. Consider all edges of $\T$ which, according to $\tilde{g_f}$, require a change. 

We now construct an extension of $f$ on $\T$ as follows: All nodes which are labelled $Y_1$ or $Y_2$ by $\tilde{g_f}$ are instead labelled $X_1$, including the leaves. We do this analogously for all other $Y_j$ and $X_i$, i.e. we replace all $Y_j$ labels by the $X_i$ label such that $Y_j \subseteq X_i$. Thus, the leaves are now labelled by $f$ and the number of edges on which changes are required may be unchanged or smaller (in case that $\tilde{g_f}$ requires a change from, say, $Y_1$ to $Y_2$, as in this case this edge would now start and end both in $X_1$). As this procedure does not introduce any new changes to edges, the score of the resulting extension $g_f$ is at least as good as that of $\tilde{g_f}$. As every MP extension of $f$ will again be at least as good as $g_f$, we altogether have $l_f(\T) \leq l_{\tilde{f}}(\T)$. This completes the proof.
\end{proof}

Next, we simplify our metric by showing that the search for a character $\tilde{f}$ which maximizes $d_{MP}(\T_1,\T_2)=\max\limits_{f}|l_f(\T_1)-l_f(\T_2)|$ can be restricted to characters which are convex on either $\T_1$ or $\T_2$. Our proof provides an explicit algorithm which, for each character $f$ with a given value of $|l_f(\T_1)-l_f(\T_2)|$, returns a character $\hat{f}$ which is convex on one of the trees and for which we have $|l_{\hat{f}}(\T_1)-l_{\hat{f}}(\T_2)|$$\geq|l_f(\T_1)-l_f(\T_2)|$. Note that $\hat{f}$ as we construct it is a refinement of $f$. In particular, it is possible that $\hat{f}$ employs strictly more character states than $f$.

\begin{theorem}\label{thm:convex} Let $\T_1, \T_2$ be two phylogenetic trees on the same set $X$ of $n$ taxa and let $f$ be an $r$-state character with $l_f(\T_i) > r-1$, i.e. $f$ is not convex on $\T_i$ for $i \in \{1,2\}$. Then, there exists an $\hat{r}$-state character $\hat{f}$ with $|l_{\hat{f}}(\T_1)-l_{\hat{f}}(\T_2)|\geq |l_f(\T_1)-l_f(\T_2)|$ and $l_{\hat{f}}(\T_2)= \hat{r} -1$ for some $\hat{r} \geq r$, i.e. $\hat{f}$ is convex on $\T_2$ and its induced parsimony distance is at least as good as the one induced by $f$. In particular, there is a character $\tilde{f}$ such that $|l_{\tilde{f}}(\T_1)-l_{\tilde{f}}(\T_2)| = d_{MP}(\T_1,\T_2)$ and $\tilde{f}$ is convex on either $\T_1$ or $\T_2$.
\end{theorem}

\begin{proof} Without loss of generality we assume that $l_f(\T_1) \geq l_f(\T_2)$. We consider a most parsimonious extension $g_f$ of $f$ on $\T_2$. Let $k$ denote the number of changes required by $g_f$. We delete all $k$ edges which require a state change according to $g_f$. Now $\T_2$ is split into $k+1$ connected components. Assume that there is a component which does not contain a leaf. This means that all edges leading to this component were edges which need a change, but this implies that $g_f$ cannot be most parsimonious: modifying $g_f$ such that all nodes in this leafless component get the same state as one of the components connected to it would reduce $k$ by at least 1. This is a contradiction and thus all components contain at least one leaf. Thus, we can consider the components labelled by the state assigned to this leaf (and note that all leaves in one component are in the same state, because all edges which require a change have been deleted). But as $f$ is not convex on $\T_2$ by assumption, $k > r-1$. We now re-label all components, i.e. we introduce $k+1$ new states and assign each component its own unique state, i.e. we label each node in this component with this state (including the leaves). This leads to a $(k+1)$-state character $\hat{f}$ and an extension $g_{\hat{f}}$ requiring exactly $k$ changes (as changes still only occur only the edges between the components once they are re-introduced to $\T_2$). Now we set $\hat{r}:=k+1$. Then, $\hat{f}$ employs $\hat{r}$ states, and $l_{\hat{f}}(\T_2)=\hat{r}-1$, i.e. $\hat{f}$ is convex on $\T_2$. 

Thus, on the one hand the number of changes needed by $\hat{f}$ equals the number of changes needed by $f$ on $\T_2$. On the other hand, $\hat{f}$ is by construction a refinement of $f$ and thus by Lemma \ref{lem:refine}, this procedure cannot decrease the parsimony score on $\T_1$, i.e. $l_f(\T_1)\leq l_{\hat{f}}(\T_1)$. So altogether, we have: $|l_{\hat{f}}(\T_1)-l_{\hat{f}}(\T_2)|\geq |l_f(\T_1)-l_f(\T_2)|$. 

Last, assume we have a character $f$ which maximizes $|l_f(\T_1)-l_f(\T_2)|$, i.e. we have $d_{MP}(\T_1,\T_2)=l_f(\T_1)-l_f(\T_2)$. Then, if $f$ is not convex on $\T_2$, we can apply the explained procedure to obtain $\tilde{f}$ such that $\tilde{f}$ is convex on $\T_2$ and $|l_{\tilde{f}}(\T_1)-l_{\tilde{f}}(\T_2)| = l_f(\T_1)-l_f(\T_2)=d_{MP}(\T_1,\T_2)$. This completes the proof.\end{proof}

Note that in order to find a character $f$ that maximizes $|l_f(\T_1)-l_f(\T_2)|$ for two given trees $\T_1, \T_2$, by Theorem \ref{thm:convex}, we can restrict our search on convex characters, but this potentially requires that the number of character states employed is not fixed. On the other hand, if we fix the number of states the optimal character might not be convex on any of the two trees (and this is why Theorem \ref{thm:convex} does not imply that MP distance is polynomial-time solveable for a fixed number of states). We elaborate this in the following lemma.

\begin{lemma}\label{lem:fixedstatenumber}
There exist trees $\T_1, \T_2$ and a fixed number $r$ of character states, such that for all characters $f$ maximizing $|l_f(\T_1)-l_f(\T_2)|$ under the restriction of employing only $r$ character states, $f$ is not convex on either $\T_1$ or $\T_2$.
\end{lemma}

\begin{proof}
We give an explicit example, taking $r=2$. Consider the two trees shown in Figure \ref{fig:nonconvex}. The set of informative binary characters that are convex on $\T_2$ is \\ $C_2:=\{ACACCCCC, CACACCCC, CCCCACAC, CCCCCACA,$\\$ AAAACCCC\}$. All characters in $C_2$ refer to splits induced by $\T_2$ and thus have a parsimony score of 1 on $\T_2$. All of the characters featuring only two taxa, i.e. characters in which two taxa are in one state and all other taxa are in another state, can at most have a parsimony score of 2 on $\T_1$. The character $AAAACCCC$, however, has a score of 1 on $\T_2$ as well as on $\T_1$. Thus, $|l_f(\T_1)-l_f(\T_2)| \leq 1$ for all convex binary characters $f \in C_2$. 

For $C_1$, which we define to be the set of all informative binary characters that are convex on $\T_1$, i.e. \\$C_1:=\{ AACCCCCC, AAACCCCC, AAAACCCC, AAAAACCC, AAAAAACC \}$, we get the same result: $|l_{\tilde{f}}(\T_1)-l_{\tilde{f}}(\T_2)| \leq 1$ for all convex binary characters ${\tilde{f}} \in C_1$. So the maximum value of the MP difference $|l_f(\T_1)-l_f(\T_2)|$ is 1 for any binary character $f$ which is convex on either $\T_1$ or $\T_2$. 

On the other hand, for the character $\hat{f}:=ACACACAC$, which is not convex on either of the trees $\T_1, \T_2$, we have $|l_{\hat{f}}(\T_1)-l_{\hat{f}}(\T_2)|=4-2=2.$ This completes the proof.

\end{proof}

\begin{figure}[ht]      \centering\vspace{0.5cm} 
       \includegraphics[width=10cm]{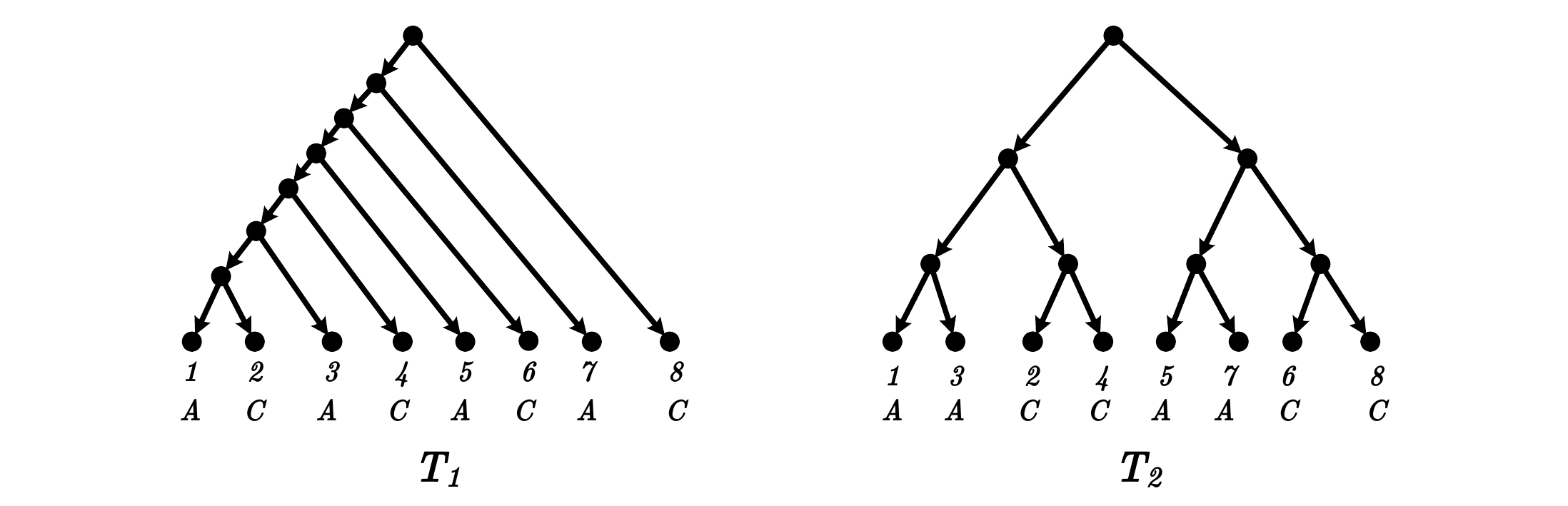}
     \caption{Two rooted binary phylogenetic $X$-trees on the same set of eight taxa and the performance of the character $ACACACAC$ on these trees. The corresponding parsimony scores are 4 and 2, respectively.}\label{fig:nonconvex}
\end{figure}

As Lemma \ref{lem:fixedstatenumber} shows, Theorem \ref{thm:convex} heavily depends on the possibility to increase the number of states employed.

We are now in a position to state and prove the following theorem, which shows that the SPR and MP distances are related.

\begin{theorem}\label{thm:lowerbound} Let $\T_1$, $\T_2$ be two binary phylogenetic $X$-trees. Then, $d_{MP}(\T_1,\T_2)\leq d_{SPR}(\T_1,\T_2)$.\end{theorem}

\begin{proof} Let $\tilde{f}:=\argmax\limits_{f} |l_f(\T_1)-l_f(\T_2)|$, i.e. $\tilde{f}$ defines $d_{MP}(\T_1,\T_2)$. Let $r=|\tilde{f}|$. By Theorem \ref{thm:convex} we may assume that $\tilde{f}$ is convex on one of the trees. So without loss of generality, we assume that $\tilde{f}$ is convex on $\T_2$, i.e. $l_{\tilde{f}}(\T_2) = r-1$. If $\T_1=\T_2$, there is nothing to show as both the SPR and MP distances are then equal to 0. Now if $\T_1 \neq \T_2$ and $\tilde{f}$ is convex on $\T_2$, we have $d_{MP}(\T_1,\T_2) = l_{\tilde{f}}(\T_1)-(r-1)$. By Theorem \ref{thm:bruenbryant} this is equal to $d_{SPR}(\T_1,\tilde{\T})$, where $\tilde{\T}$ is the tree which minimizes $\min\limits_{\T}d_{SPR}(\T_1,\T)$ such that $\tilde{f}$ is convex on $\T$. So $d_{MP}(\T_1,\T_2)=d_{SPR}(T_1,\tilde{\T})\leq d_{SPR}(\T_1,\T_2)$, which completes the proof.
\end{proof}

Theorem \ref{thm:lowerbound} states that the MP distance provides a lower bound to the SPR distance. It does not, however, prove that the two distance measures are actually different. But this becomes apparent when considering the two unrooted trees $\T_1=((1,2),3,(4,5))$ and $\T_2=((1,4),5,(2,3))$ as depicted in Figure \ref{fig:MPnotSPR}. The SPR distance between these two trees is 2, as can be seen when considering all 7
edges of, say, $\T_1$: no matter which one we prune and regraft to another place, we cannot generate $\T_2$. If we, on the other hand, first cut leaf $4$ and attach it next to leaf $1$ and suppress all resulting vertices of degree 2, we get the tree $((1,4),2,(3,5))$. Pruning the edge leading to leaf 5 and regrafting it between the cherry $(1,4)$ and leaf $2$, we obtain tree $\T_2$. So the SPR distance from $\T_1$ to $\T_2$ is 2, but an exhaustive search through all possible informative characters on 5 taxa shows that the MP distance in this case is 1. It is, for instance, achieved by character $AACCC$, which has a score of 1 on $\T_1$ and 2 on $\T_2$. So in this case, the MP distance is strictly smaller than the SPR distance.

\begin{figure}[ht]      \centering\vspace{0.5cm} 
          \includegraphics[width=10cm]{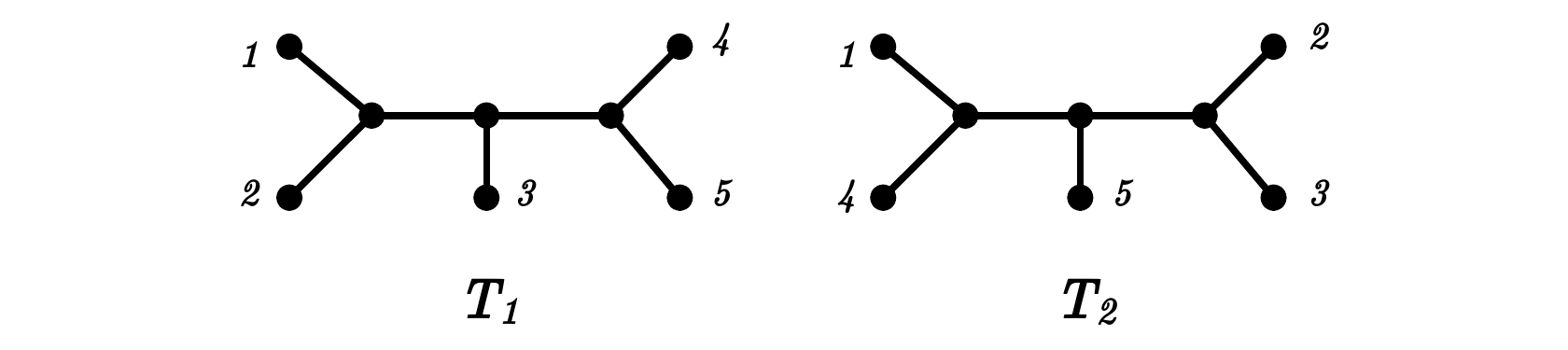}
        \caption{The MP distance of these two trees is 1, but their SPR distance is 2.}\label{fig:MPnotSPR}
\end{figure}

Note that by \citep[Proposition 5.1]{bordewich_semple_2004}, the rooted SPR distance $d_{rSPR}$ and the (unrooted) SPR distance $d_{SPR}$ are closely related. We present this proposition in the following lemmas and subsequently combine them in order to match our purposes. This way, we justify the fact that we only consider the (unrooted) SPR distance in the following, as any lower bound for the latter also provides a lower bound for the rooted SPR distance.

\begin{lemma}[Bordewich and Semple, 2004, Proposition 5.1]\label{lem:bordewich1} Let $\T_1$ and $\T_2$ be two rooted binary phylogenetic $X$-trees. Let $\T_1'$ and $\T_2'$ be the (unrooted) binary phylogenetic $X \cup \{r\}$-trees obtained by attaching a pendant leaf $r$ to the root of $\T_1$ and $\T_2$, respectively, and then regarding the resulting trees as unrooted. Then, $d_{SPR}(\T_1',\T_2') \leq d_{rSPR}(\T_1,\T_2)$.
\end{lemma}

Note that Lemma \ref{lem:bordewich1} is not directly applicable to our setting, because trees $\T_1'$ and $\T_2'$ as mentioned in this lemma have $n+1$ leaves, where $n=|X|$, whereas what we want is to relate a rooted binary tree to the same tree regarded as unrooted by suppressing the root node, i.e. by deleting the degree 2 node and its incident edges and reconnecting the then unconnected components with a new edge. We achieve this in the following lemma.

\begin{lemma} \label{lem:SPRrSPR} Let $\T_1$ and $\T_2$ be two rooted binary phylogenetic $X$-trees. Let $\tilde{\T_1}$ and $\tilde{\T_2}$ be the unrooted binary phylogenetic $X$-trees derived from $\T_1$ and $\T_2$, respectively, by suppressing the root. Let $\T_1'$ and $\T_2'$ be the (unrooted) binary phylogenetic $X \cup \{r\}$-trees obtained by attaching a pendant leaf $r$ to the root of $\T_1$ and $\T_2$, respectively, and then regarding the resulting trees as unrooted. Then, $d_{SPR}(\tilde{\T_1},\tilde{\T_2})\leq d_{SPR}(\T_1',\T_2') \leq  d_{SPR}(\tilde{\T_1},\tilde{\T_2})+1$.
\end{lemma}

\begin{proof} First, note that $\tilde{\T_1}$ is a subtree of $\T_1'$ and $\tilde{\T_2}$ is a subtree of $\T_2'$.  Assume $d_{SPR}(\tilde{\T_1},\tilde{\T_2})=1$, i.e. one SPR move is needed to get from $\tilde{\T_1}$ to $\tilde{\T_2}$. Then also at least one SPR move is needed to get from $\T_1'$ to $\T_2'$ to get the correct arrangement of the subtree $\tilde{\T_1}$ or $\tilde{\T_2}$, respectively (and possibly one more if the position of $r$ also has to be modified). So in this case, $d_{SPR}(\tilde{\T_1},\tilde{\T_2})\leq d_{SPR}(\T_1',\T_2') $. However, if more SPR moves are needed to get from $\tilde{\T_1}$ to $\tilde{\T_2}$, we consider a shortest path of single SPR moves $\tilde{\T_1}, \T_a,T_b,\ldots,\tilde{\T_2}$ from $\tilde{\T_1}$ to $\tilde{\T_2}$. Each such required move enforces the same move from $\T_1'$ to $\T_2'$. In the end, possibly an adjustment of the position of leaf $r$ has to be made, but this can be made in a single move cutting leaf $r$ and attaching it at the appropriate position. So iteratively, we get $d_{SPR}(\tilde{\T_1},\tilde{\T_2})\leq d_{SPR}(\T_1',\T_2') \leq  d_{SPR}(\tilde{\T_1},\tilde{\T_2})+1$.  \end{proof} 

Now we are finally in a position to show the required relation between $d_{MP}$, $d_{SPR}$ and $d_{rSPR}$.

\begin{theorem}\label{thm:boundrSPR} Let $\T_1$ and $\T_2$ be two rooted binary phylogenetic $X$-trees. Let $\tilde{\T_1}$ and $\tilde{\T_2}$ be the unrooted binary phylogenetic $X$-trees derived from $\T_1$ and $\T_2$, respectively, by suppressing the root. Then, $d_{MP}(\T_1,\T_2)=d_{MP}(\tilde{\T_1},\tilde{\T_2})\leq d_{SPR}(\tilde{\T_1},\tilde{\T_2})$\\$ \leq d_{rSPR}(\T_1,\T_2)$.
\end{theorem}

\begin{proof} The equality follows from the fact that for parsimony there is no difference between rooted and unrooted trees. The first inequality follows from Theorem \ref{thm:lowerbound}. It remains to show that $d_{SPR}(\tilde{\T_1},\tilde{\T_2}) \leq d_{rSPR}(\T_1,\T_2)$. Let $\T_1'$ and $\T_2'$ be the two $X\cup \{r\}$-trees as defined in Lemma \ref{lem:bordewich1}. Then Lemma \ref{lem:SPRrSPR} gives $d_{SPR}(\tilde{\T_1},\tilde{\T_2})\leq d_{SPR}(\T_1',\T_2')$ and Lemma \ref{lem:bordewich1} leads to $d_{SPR}(\T_1',\T_2') \leq d_{rSPR}(\T_1,\T_2)$. So altogether we have $d_{SPR}(\tilde{\T_1},\tilde{\T_2})\leq d_{rSPR}(\T_1,\T_2)$. This completes the proof. 
\end{proof}

So Theorem \ref{thm:boundrSPR} shows that the Maximum Parsimony distance is a lower bound both for the rooted and the unrooted SPR distance, but also that the unrooted SPR distance itself is a lower bound for the rooted SPR distance, which is why analyzing the relationship of $d_{MP}$ to $d_{SPR}$ is sufficient. 

We finish this section by noting that each character provides a lower bound for the SPR distance. 

\begin{corollary}\label{cor:lowerbound} Let $\T_1$, $\T_2$ be two binary phylogenetic $X$-trees and let $f$ be any character on $X$. Then, $|l_f(\T_1)-l_f(\T_2)|\leq d_{SPR}(\T_1,\T_2)$.\end{corollary}

\begin{proof} This follows from Theorem \ref{thm:lowerbound} and the fact that  $d_{MP}(\T_1,\T_2) \geq |l_f(\T_1)-l_f(\T_2)|$ by definition of $ d_{MP}(\T_1,\T_2)$. \end{proof}

\subsubsection{Bounds on the MP distance}\label{sec:bounds}

In this section, we want to provide an upper bound for the MP distance. We start with the following lemma.

\begin{lemma} \label{lem:badchar}The MP score of an $r$-state character on taxon set $X$ on any phylogenetic $X$-tree is at most $\lfloor(r-1)\cdot \frac{n}{r} \rfloor$, where $|X|=n$. 
\end{lemma}

\begin{proof}
It is easy to see that the parsimony score of the star tree is never better than that of any refinement, as in the star tree one state will be the root state and all leaves which are not in this state require exactly one change (note that one can construct binary trees with this property, too, but for simplicity we now consider the star tree). Now let us consider a character $\tilde{f}$ using the states $c_1,\ldots,c_r$ and let $n_i$ denote the number of leaves assigned state $c_i$, for $i=1,\ldots,r$. Note that $\sum_{i=1}^r n_i = n$. Without loss of generality assume $n_1\geq n_i$ for all $i=2,\ldots,r$. As parsimony seeks to minimize the number of changes needed, it will choose the state which occurs most often as root state, so the root will be in state $c_1$. In the extreme case of a star tree, the number of changes needed by any most parsimonious extension is $\sum_{i=2}^r n_i$, which is maximized when $n_1=n_i$ for all $i=2,\ldots,r$. Note that this choice is only possible if $n=r\cdot n_1$, in which case the MP score will be exactly $(r-1) \cdot \frac{n}{r}$. In the case that $n$ is not a multiple of $r$, taking $n_1 = \lceil \frac{n}{r} \rceil$ yields a maximum MP score
of $n -  \lceil \frac{n}{r} \rceil$, which is  $\lfloor(r-1)\cdot \frac{n}{r} \rfloor$, completing the proof.

\end{proof}

We now use Lemma \ref{lem:badchar} to state the following lemma, which in turn will then provide the desired bound on the MP distance.

\begin{lemma} \label{lem:theobound}Let $\T_1$, $\T_2$ be two (rooted or unrooted) phylogenetic $X$-trees with $|X|=n$. Let $f$ be any $r$-state character on $X$ for some $r \leq n$. Then, $|l_f(\T_1)-l_f(\T_2)|\leq $ $\lfloor(r-1)(\frac{n}{r}-1)\rfloor$. \end{lemma}

\begin{proof} Without loss of generality we assume $l_f(\T_2)\leq l_f(\T_1)$. First we note that the parsimony score of $f$ on any tree, and thus particularly on $\T_2$, is at least $r-1$, because only one of the states employed by $f$ can be the root state, and to all other states there has to be at least one change (and in the optimal case, when the score on some tree equals exactly $r-1$, $f$ is convex on this tree). Then by Lemma \ref{lem:badchar}, the score of $f$ on $\T_1$ cannot exceed $\lfloor(r-1)\frac{n}{r}\rfloor$. So altogether, $|l_f(\T_1)-l_f(\T_2)|\leq \lfloor(r-1)\frac{n}{r}\rfloor - (r-1) =\lfloor(r-1)\frac{n}{r}\rfloor - \lfloor r-1\rfloor$\\$ \leq \lfloor (r-1)\frac{n}{r}-(r-1) \rfloor = \lfloor(r-1)(\frac{n}{r}-1)\rfloor$. The latter inequality is due to $ r \in \NN$. This completes the proof. 
\end{proof}

We now derive an upper bound on the MP distance between two phylogenetic trees using Lemma \ref{lem:theobound}. 

\begin{theorem}\label{thm:theobound} Let $\T_1$, $\T_2$ be two phylogenetic $X$-trees with $|X|=n$. Then, \\ $d_{MP}(\T_1,\T_2)\leq n - 2\sqrt{n} + 1$.
\end{theorem}

\begin{proof} We define the function $f(r,n):= (r-1)\left( \frac{n}{r} -1\right)$. By considering the first and second derivative with respect to $r$, one can see that $f(r,n)$ is maximized at $r=\sqrt{n}$. Let $\tilde{f}$ be the character which maximizes $|l_f(\T_1)-l_f(\T_2)|$ and thus provides $d_{MP}(\T_1,\T_2)=|l_{\tilde{f}}(\T_1)-l_{\tilde{f}}(\T_2)|$. Let $\tilde{r}$ be the number of states employed by $\tilde{f}$. Then by Lemma \ref{lem:theobound}, we conclude $d_{MP}(\T_1,\T_2)=|l_{\tilde{f}}(\T_1)-l_{\tilde{f}}(\T_2)|\leq \lfloor(\tilde{r}-1)( \frac{n}{\tilde{r}} -1)\rfloor \leq (\tilde{r}-1)\left( \frac{n}{\tilde{r}} -1\right) = f(\tilde{r},n)\leq f(\sqrt{n},n) = (\sqrt{n}-1)^2=n - 2\sqrt{n} + 1.$ This completes the proof.
\end{proof}

Theorem \ref{thm:theobound} is useful in the sense that it provides an upper bound on $d_{MP}$, which is in fact a tight bound as can be seen in Figure \ref{fig:boundreached}: There, the trees $\T_1=((((1,2),3),(4,(5,6))),(7,(8,9)))$ and $\T_2=((((1,4),7),(2,(5,8))),$\\ $(3,(6,9)))$ are depicted together with the character $f=ACGACGACG$. The parsimony score of $f$ on $\T_1$ and $\T_2$ can be easily calculated with the Fitch algorithm \citep{fitch_1971} to be $l_f(\T_1)=6$ and $\l_f(\T_2)=2$. In particular, $f$ is convex on $\T_2$. By Theorem \ref{thm:theobound}, the upper bound of the MP distance two 9-taxon trees is $n - 2\sqrt{n} + 1=9-2\cdot 3 + 1= 4$. As $l_f(\T_1)-\l_f(\T_2)=4$, this implies that $f$ provides the MP distance and that the optimum is indeed achieved. Note that this theoretical bound is achieved here even though both trees under consideration are binary -- but examples for multifurcating trees can be constructed in a similar way. In fact, if we chose $\T_1$ to be the star tree, we would get the same result. 

\begin{figure}[ht]      \centering\vspace{0.5cm} 
  \includegraphics[width=10cm]{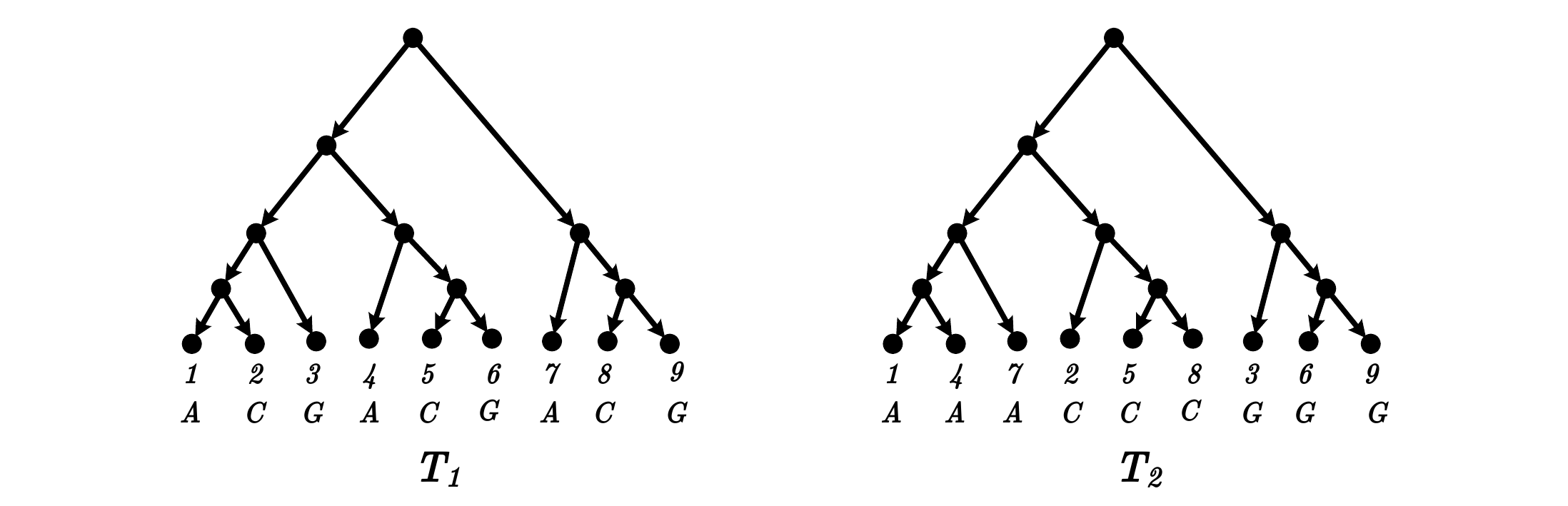}
    \caption{Two rooted binary phylogenetic $X$-trees $\T_1$ and $\T_2$ for $X=\{1,\ldots,9\}$ and the performance of the character $f=ACGACGACG$ on these trees. The corresponding parsimony scores are $l_f(\T_1)=6$ and $l_f(\T_2)=2$, respectively. The MP distance $d_{MP}(\T_1,\T_2)$ can be shown to equal 4. It is achieved by $f$.}\label{fig:boundreached}

\end{figure}

\subsubsection{The number of character states needed to maximize the parsimony difference between two trees}\label{sec:statenumber}

We continue by considering the obvious question whether or not the MP distance can be maximized by considering only characters with a fixed number of states. For instance, one could wonder if the optimum performance difference can always be reached by a binary character. The following theorem leads to the conclusion that this is unfortunately not the case.

\begin{theorem} \label{thm:noconstantstatenumber}
Let $r >1$. Then there exist two binary phylogenetic trees $\T_1, \T_2$ on a set of $n=(r+1)^2$ taxa and an $(r+1)$-state character $f$ such that $|l_f(\T_1)-l_f(\T_2)| > |l_{\hat{f}}(\T_1)-l_{\hat{f}}(\T_2)|$ for all $r$-state characters $\hat{f}$.
\end{theorem}

\begin{proof} We give an explicit construction of $\T_1$ and $\T_2$ as rooted binary trees, but for unrooted trees the root can later on be ignored. We draw a root node and continue by adding a subtree of size $r+1$ taxa to the right hand side of the root. This subtree may have any binary topology, for instance it can be chosen to be a so-called caterpillar. On the left hand side of the root, we draw an edge leading to a new node, which we regard the root of a new subtree. With this new root we continue just as before by adding another subtree of size $r+1$ to the right. We continue this procedure until there are $r$ such subtrees to the right of their corresponding parental node. For the last such node we create another subtree of size $r+1$ taxa and add it to the left hand side of this node in a way that the entire tree remains binary. Altogether, there are now $r+1$ subtrees of $r+1$ taxa each. 

Now for $\T_1$, label the leaves $1,2,3,...,(r+1)^2$ from the left to the right. Thus, the leftmost subtree will have the labels $1,...,r+1$, the second subtree $r+2, r+3,..., 2(r+1)$, and so on. For $\T_2$, keep the same tree topology of $r+1$ subtrees, but use a different leaf labelling: Label the leftmost subtree $1, (r+1)+1, 2(r+1)+1,..., k (r+1)+1$, the second subtree $2, (r+1)+2, 2(r+1)+2,...,k(r+1)+2$, and so on. The last subtree will be labelled $r+1, (r+1)+(r+1), 2(r+1)+(r+1),...,r(r+1)+(r+1)$. For $r=2$, this will lead to the trees shown in Figure \ref{fig:boundreached}. 

Now introduce character $f:= \underbrace{c_1,c_2,c_3,...,c_{r+1}}_{(r+1) times}$ \\$= c_1,c_2,c_3,...,c_{r+1},c_1,c_2,c_3,...,c_{r+1},\ldots$ for some distinct $r+1$ character states \\ $c_1,\ldots,c_r,c_{r+1}$. Since in every pending $(r+1)$-taxon subtree of $\T_1$ all $r+1$ character states appear, the score of this character on each of these subtrees is $r$ and the total MP score on $\T_1$ is $l_f(\T_1)= (r+1)r$. According to the labelling of $\T_2$, $f$ is convex there as the first pending $(r+1)$-taxon subtree only contains taxa that are in state $c_1$, the second one only taxa of state $c_2$, and so on. Thus, by construction, the MP score on $\T_2$ is $l_f(\T_2)=(r+1)-1=r$. Altogether, the difference equals $|l_f(\T_1)- l_f(\T_2)| = (r+1)r-r=r^2$.

Next we show that for every $r$-state character $\hat{f}$, the difference $|l_{\hat{f}}(\T_1)- l_{\hat{f}}(\T_2)|$ is smaller than $r^2$. Any $r$-state character has a parsimony score of at least $r-1$ on any phylogenetic tree. Furthermore, on the two given trees the maximum MP score, say on $\T_1$, will be obtained by a character that maximizes the score on each of the pending $(r+1)$-size subtrees. This can be achieved by using all $r$ character states on each of these subtrees. Since they all have $r+1$ leaves, this means that on every subtree one character state will appear twice and the score obtained in any of these subtrees is $r-1$. If a different root state is suggested by MP for $r$ of the $r+1$ pending subtrees, on all but two edges leading to these subtrees a substitution will be suggested by MP. Then, the MP score is $(r+1)(r-1) + (r-2)$. So altogether, the maximum difference is at most $(r+1)(r-1)+(r-2)-(r-1)=(r-1)(r+1)=r^2-2 <r^2$. This completes the proof.
\end{proof}

Note that Figure \ref{fig:boundreached} illustrates the construction described in the proof of Theorem \ref{thm:noconstantstatenumber} for the case where $r=2$: The two trees depicted here have an MP distance of at least 4, which can be seen by considering the depicted 3-state character. However, by Lemma \ref{lem:theobound}, for all binary characters $\hat{f}$ on nine taxa we obtain a difference $|l_{\hat{f}}(\T_1)-l_{\hat{f}}(\T_2)|\leq \lfloor(2-1)(\frac{9}{2}-1)\rfloor=3$. This means that no binary character reaches a difference of $4$ for those two trees. 

So there exist two trees where more than two states can give a higher difference, and in general, by Theorem \ref{thm:noconstantstatenumber}, no constant number of character states is sufficient as the described construction can be extended to more taxa. Thus, the optimal number of states to employ depends on the tree shapes of the trees under consideration. This can be seen when considering the theoretical bound of 3 for binary characters on nine taxa: For the two trees displayed in Figure \ref{fig:boundreached}, not even the bound of 3 can be reached, as an exhaustive search reveals that there is no binary character $f$ such that $l_f(\T_1)=4$ and at the same time $\l_f(\T_2)=1$ or vice versa.

All this already gives a hint to the complexity of the underlying problem, which we analyze further in the subsequent section.

\section{On the complexity of calculating the MP distance}\label{sec:complexity}

In this section we show that computation of MP distance is NP-hard on non-binary trees, and that
a fixed-state variant of the problem is also NP-hard. As we shall see the NP-hardness reductions do
not, in their present form, work for binary trees. We address this issue further in Section \ref{sec:discussion}.

 We begin, however,
with a positive result which introduces several key concepts used by the more involved hardness results.

\begin{lemma}
Let $\T_1$ and $\T_2$ be two (not necessarily binary) phylogenetic $X$-trees with $|X|=n$, where at least
one of the trees is a star tree. Then  $d_{MP}(\T_1,\T_2)$ can be computed in polynomial time.
\end{lemma}
\begin{proof}
If $\T_1$ and $\T_2$ are identical then $d_{MP}(\T_1,\T_2)=0$ and we are done, so let us assume that
$\T_1$ is a star, $\T_2$ is not a star and  $d_{MP}(\T_1, T_2) > 0$. Without loss of generality we also assume that
$\T_1$ and $\T_2$ are rooted. Clearly, $\T_2$ is a refinement of $\T_1$, so for any character $f$ we have by Lemma \ref{lem:refine} $l_f(\T_1) \geq l_f(\T_2)$. If we combine this with Theorem \ref{thm:convex} then we know that
there exists a character $\tilde{f}$ such that $l_{\tilde{f}}(\T_1) -  l_{\tilde{f}}(\T_2) = d_{MP}(\T_1, T_2)$ and
$\tilde{f}$ is convex for $\T_2$.  Let $r(\tilde{f})$ be the number of states in $\tilde{f}$.
The value $l_{\tilde{f}}(\T_1)$ is equal to $n-m(\tilde{f})$, where $m(\tilde{f})$ is the frequency of the most frequently used state in $\tilde{f}$. Hence $d_{MP}(\T_1, \T_2) =  (n-m(\tilde{f})) - (r(\tilde{f})-1) = (n+1)-(m(\tilde{f})+r(\tilde{f}))$. To construct such an
$\tilde{f}$ it is sufficient to construct a character $f$ that is convex for $\T_2$ and such that $m(f)+r(f)$ is
minimized. We can do this by exhaustively trying all possible pairs $(m,r)$, where $1 \leq m, r \leq n$, and returning the minimum value of $m+r$ ranging over all ``valid'' pairs $(m,r)$. A pair $(m,r)$ is valid if there exists a character
$f$ that is convex for $\T_2$ such that $m(f) \leq m$ and $r(f) \leq r$. Due to minimality it is sufficient to consider
$(m,r)$ pairs in order of increasing $m+r$ and to stop as soon as a valid pair is encountered. The only task that remains is
to determine validity in polynomial time. To do this we first observe that, if a character $f$ has exactly $r$ states
and is convex for $\T_2$, then any optimal extension of $f$ to the interior nodes of $\T_2$ naturally induces
$r-1$ edges upon which mutations between character states occur. If these $r-1$ edges are removed, $\T_2$ is partitioned into exactly $r$ components, and $m(f)$ is then equal to the maximum number of taxa in any
of these components. Hence, for a given value $m$, determining the smallest value $r$ such that $(m,r)$ is
valid, is equivalent to the question: what is the smallest number of edge cuts I need to make to $\T_2$ to ensure
that every resulting component contains at most $m$ taxa? Fortunately this problem can be solved in polynomial time by giving each taxon weight 1, each inner node weight 0, and using the polynomial-time ``tree partitioning'' algorithm in \cite{treepartitioning} to compute an optimal $m$-partition. $\Box$ 
\end{proof}

Let $d_{MP}^{i}(\T_1,\T_2)$ denote the MP distance of two trees when restricted to characters with at most $i$ states. Note that, for constant $i$,
there is no obvious relationship between the complexity of computing $d_{MP}^{i}(\T_1,\T_2)$
and $d_{MP}(\T_1,\T_2)$. However, as we will now show, both problems are NP-hard.
We start with the hardness proof for $d_{MP}^{2}(\T_1,\T_2)$. The proof that
$d_{MP}(\T_1,\T_2)$ is NP-hard will use similar, but somewhat more complex, techniques. First, we present some auxiliary results.

\begin{observation}
\label{obs:mostone} 
Let $f$ be a character on $X$ and $\T$ a tree on $X$. Let $f'$ be any character obtained from $f$ by changing the state of exactly one taxon. Then
$l_f(\T) - 1 \leq l_{f'}(\T) \leq l_f(\T) + 1$ i.e. the parsimony score can change by at most one.
\end{observation}
\begin{proof}
Suppose $l_{f'}(\T) \leq l_{f}(\T) - 2$. Consider any extension of $f'$ to the interior
nodes of $\T$ that achieves $l_{f'}(\T)$ mutations. Using the same extension but on $f$
gives at most $l_{f'}(\T)+1$ mutations, because only one taxon changed state. So
$l_{f}(\T) \leq l_{f'}(\T)+1 \leq l_{f}(\T)-1$, which is a contradiction. In the other direction,
take any optimal extension of $f$ and apply it to $f$'. At most one extra mutation will be created, so $l_{f'}(\T) \leq l_{f}(\T) + 1$.
\end{proof}

\begin{lemma}
\label{lem:monochrome}
Let $f$ be an optimal character for two trees $\T_1$ and $\T_2$ i.e. $d_{MP}(\T_1, \T_2) = |l_f(\T_2) - l_f(\T_1)|.$
Without loss of generality assume $l_f(\T_1) < l_f(\T_2)$. Then we can construct in polynomial time an optimal character $f'$
with the
following property: $l_{f'}(\T_1) < l_{f'}(\T_2)$ and for each vertex $u$ of $\T_1$ such that all $u$'s children are leaves, $f'$ assigns all the children of $u$ the same state. 
\end{lemma}
\begin{proof}
Consider a vertex $u$ of $\T_1$ such that all its children are taxa, but such that $f$ assigns
the children two or more states. We calculate an optimal extension of $f$ to the interior
nodes of $\T_1$ by applying Fitch's algorithm. Fitch will allocate a most frequently
occuring state amongst the children of $u$, to $u$. (If there is a unique such state then
Fitch will choose it, otherwise it will break ties in the top-down phase of the algorithm). Let
$s$ be the state allocated to $u$. Choose a child of $u$ that does not have state $s$ and
change its state to $s$. This yields a new character $f^{*}$. Clearly, $l_{f^{*}}(\T_1) < l_{f}(\T_1)$,
simply by using the same extension that Fitch gave. Combining this with Observation \ref{obs:mostone} gives
$l_{f^{*}}(\T_1) = l_{f}(\T_1) - 1$ and thus $l_{f^{*}}(\T_2) = l_{f}(\T_2)-1$ (otherwise $f$
could not have been optimal). Hence, $f^{*}$ is also an optimal character, and
$l_{f^{*}}(\T_1) < l_{f^{*}}(\T_2)$. This process can be iterated as long as necessary
until all the children of $u$ have the same state. (Note that in subsequent iterations
Fitch will definitely assign state $s$ to $u$, because $s$ will have become the unique
most frequently occurring state amongst the children of $u$). Then we can iterate
the process on other vertices $u'$ whose children do not all have the same state,
for as long as necessary. Polynomial time is guaranteed since the state of each taxon is changed at most once.
\end{proof}

\begin{observation}
\label{obs:boundedmonochrome}
Lemma \ref{lem:monochrome} also holds for optimal characters under the $d_{MP}^{i}(\T_1, \T_2)$ model.
\end{observation}
\begin{proof}
The transformation in the proof of Lemma \ref{lem:monochrome} does not increase the number of states in the
character.
\end{proof}


\begin{figure}[ht]      \centering\vspace{0.5cm} 
       \includegraphics[width=13cm]{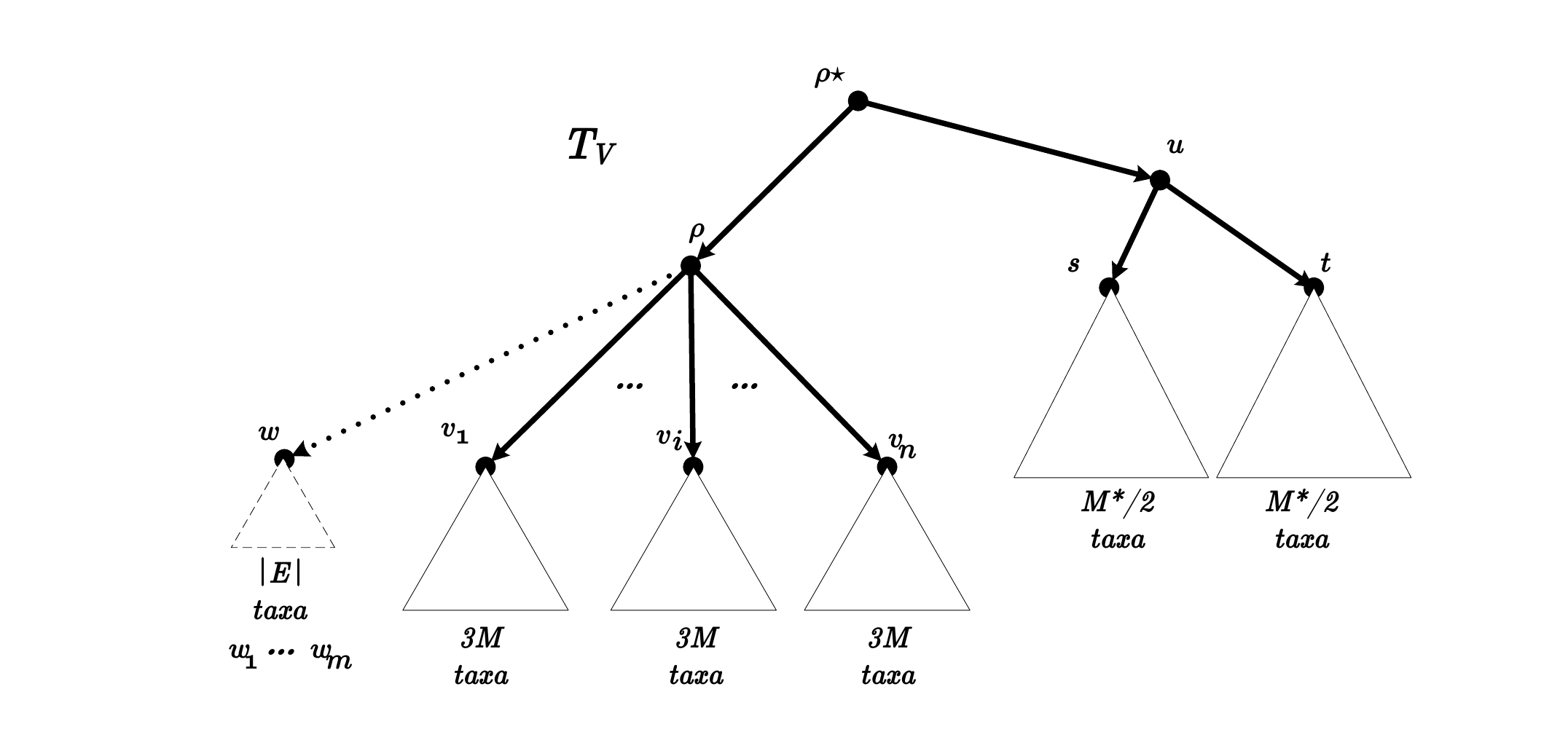}
\caption{The tree $\T_V$ used in the proofs of Lemma \ref{lem:boundedhard} and Theorem \ref{thm:mainhard}. Note that the dotted parts are only used in the proof of Theorem \ref{thm:mainhard}.}
\label{fig:tv}
\end{figure}



\begin{lemma}
\label{lem:boundedhard}
Computing $d_{MP}^{2}(\T_1,\T_2)$ is NP-hard.
\end{lemma}
\begin{proof}
We reduce from the NP-hard (and APX-hard) problem \textsc{CUBIC MAX CUT}, see \cite{alimonti2000}. In this problem we
are given an undirected 3-regular graph $G=(V,E)$. A \emph{cut} is a bipartition of $V$ and the size of the cut is the number of edges that cross the bipartition. The goal is
to compute a cut of maximum-size. The restriction to 3-regular graphs is not strictly
necessary but simplifies the proof somewhat. Clearly, $|E| = 3|V|/2$. Let $MAXCUT(G)$ be
the size of the maximum cut in $G$; it is well-known that $MAXCUT(G) \geq 2|E|/3 = |V|$ (by repeatedly moving nodes to the other side of the partition that have only one of their three incident edges in the cut). Let
$\{0,1\}$ be the two character states. During this proof we will write ``any character''
as shorthand for ``any character with at most two states''.

The high-level idea is to construct two trees $\T_1$ and $\T_2$, henceforth
referred to as $\T_V$ and $\T_E$, where $\T_V$ encodes the vertices and $\T_E$ the
edges of $G$. In $\T_V$ the character states $\{0,1\}$ will be used to indicate whether
a vertex is on the left or right side of the bipartition. The mutations induced in $\T_E$ will
be used to count the number of edges crossing the bipartition. Intuitively,
$d_{MP}^{2}(\T_V,\T_E)$ will be maximized by choosing a maximum-size cut.

Let $V = \{v_1, \ldots, v_n\}$ and $E = \{e_1, \ldots, e_m\}$. Throughout the reduction we will utilize two large numbers,
$M$ and an even number $M^{*}$, such that $n << M << M^{*}$ but such that
both are still at most $poly(n)$. In due course we will explain
how these numbers are calculated.  Both trees will have
$3|V|M + M^{*}$ taxa.

\begin{figure}[ht]      \centering\vspace{0.5cm} 

\includegraphics[width=12cm]{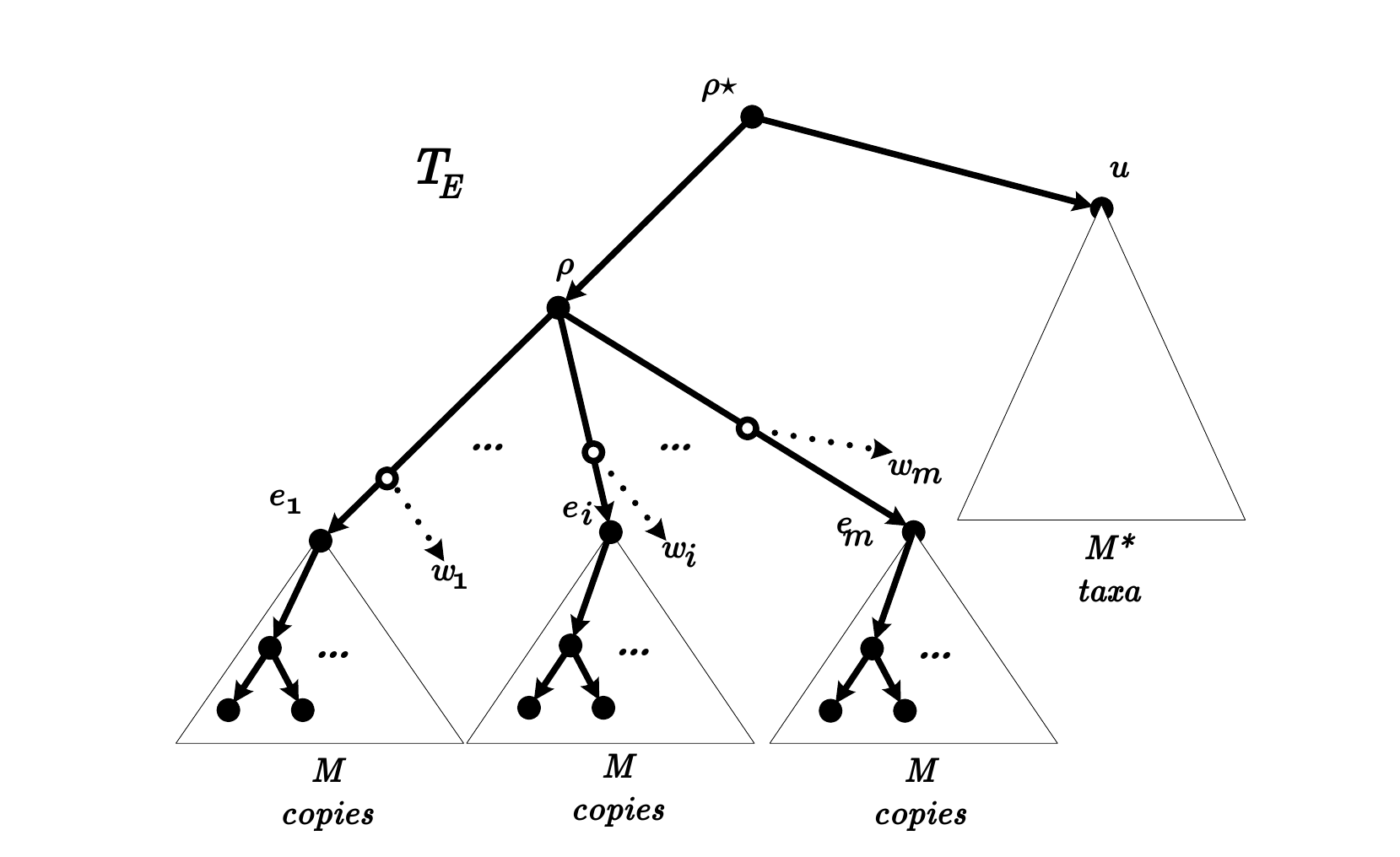}
       
   \caption{The tree $\T_E$ used in the proofs of Lemma \ref{lem:boundedhard} and Theorem \ref{thm:mainhard}. Note that the dotted parts and the unfilled vertices are only used in the proof of Theorem \ref{thm:mainhard}. Below each $e_i$ vertex there are in total $2M$ taxa, organized in cherries. If $e_i$ has endpoints $v_a$ and $v_b$ in $G$, then each
cherry will have one taxon from the $v_a$ clade in $\T_V$ and one taxon from the $v_b$ clade.}
\label{fig:te}
\end{figure}

To construct $\T_V$ we first introduce vertices $\rho^{*}, \rho, s, t, u$ and $\{v_1, \ldots, v_n\}$. We add edges $\{\rho^{*},\rho\}$, $\{\rho^{*},u\}$, $\{u,s\}$, $\{u,t\}$ and
$\{\rho, v_i\}$ for $1 \leq i \leq n$. We connect each $v_i$ to $3M$ taxa. We connect
$s$ to $M^{*}/2$ taxa and $t$ to $M^{*}/2$ taxa. Figure \ref{fig:tv} depicts the construction idea.

To construct $\T_E$ we introduce vertices $\rho^{*}, \rho, u$ and $\{e_1, \ldots, e_m\}$.
We add edges $\{\rho^{*},\rho\}$, $\{\rho^{*},u\}$ and
$\{\rho, e_i\}$ for $1 \leq i \leq m$. To $u$ we connect the $M^{*}$ taxa that were
connected to $s$ and $t$ in $\T_V$. Next we introduce $mM$ vertices $e_{i,j}$ for
$i \in \{1,\ldots,m\}$ and $j \in \{1,\ldots,M\}$. We connect each vertex $e_i$ to
all the $e_{i,j}$, for $j \in \{1,\ldots,M\}$. To each $e_{i,j}$ we connect two taxa, representing the endpoints (of the $j$th copy) of edge $e_i$. To determine
which two taxa these are, suppose in $G$ the edge $e_i$ is connected to vertices $v_a$ and
$v_b$. Then one of the two taxa is taken from the clade of taxa we connected to $v_a$ in $\T_V$, and the other from the clade of taxa beneath $v_b$ in $\T_V$. The exact mapping
chosen does not matter. The construction is depicted by Figure \ref{fig:te}. This completes the construction.

For $x \in \mathbb{R}$ let $round(x)$ be $x$ rounded to the nearest integer. (By
construction we will actually only use values of $x$ that are at most 1/3 above or below the nearest integer; this will become clearer later). We shall
prove the following:
\begin{equation}
MAXCUT(G) = round \bigg ( \frac{ d_{MP}^{2}(\T_V,\T_E) - M^{*}/2 }{M} \bigg )
\label{eq:round}
\end{equation}
Moreover, we will show how any optimal character $f$ can
be transformed in polynomial-time into a maximum-size cut of $G$. From this the NP-hardness
of computing $d_{MP}^{2}(\T_V,\T_E)$ will follow.

Note that in $\T_V$ the subtree rooted at $u$ is a refinement of the subtree rooted at $u$ in $\T_E$. This means that, for any character $f$, a (crude) upper bound on $l_f(\T_V) - l_f(\T_E)$
is $3|V|M$. To see this, observe that for any character the number of mutations incurred in $\T_V$ in the
subtree rooted at $u$, is less than or equal to the number of mutations incurred in the corresponding subtree of $\T_E$.
Therefore mutations in this subtree can never contribute to an increase in $l_f(\T_V) - l_f(\T_E)$. 
Hence, an upper bound on $l_f(\T_V) - l_f(\T_E)$ can be
achieved by maximizing the number of mutations that occur in $\T_V$ on edges that are not in this subtree, and minimizing the number of mutations that occur in $\T_E$ on edges that are not in this subtree. A (trivial) lower bound on the latter is 0, while a (trivial) upper bound on the former is $3|V|M$ (e.g. by assigning the same state to all internal nodes of $\T_V$).

On the other hand, consider a character $f$ such that the $M^{*}/2$ taxa
underneath $s$ in $\T_V$ are allocated state 0, and all other taxa are allocated state 1. Then
$l_f(\T_V)=1$ and $l_f(\T_E) \geq M^{*}/2.$ For this reason we choose $M^{*}$ such that
\[
M^{*}/2 - 1 > 3|V|M.
\]
After choosing $M^{*}$ this way we know that $d_{MP}^{2}(\T_V,\T_E) \geq  (M^{*}/2)- 1$ and, more importantly, that for every optimal character $f$, $l_f(\T_V) < l_f(\T_E)$.

From Observation \ref{obs:boundedmonochrome} we can therefore assume that, for each vertex of $\T_V$ in
the set $\{s,t,v_1, \ldots, v_n\}$, all the taxa beneath the vertex are allocated the same state by $f$. (Note also
that the state of the taxa in the $s$ clade, and the state of the taxa in the $t$ clade, must be different, otherwise
$f$ could not possibly be optimal.)

Such a character naturally induces a cut, with
a vertex $v_i$ being on the left (right) side of the bipartition if the taxa in its clade
are allocated state 0 (1). The core observation is that if an edge $e_i$ is in the induced cut
(i.e. crosses the bipartition) then $f$ will induce at least $M$ mutations in the subtree of
$\T_E$ rooted at $e_i$. On the other hand, an edge $e_i$ \emph{not} in the cut will induce 0 mutations in the subtree rooted at $e_i$.

More formally, suppose $f$ induces a cut of size $k$. The parameter $k$ does not, in itself,
give us enough information to \emph{exactly} determine $l_f(\T_V)$ and $l_f(\T_E)$, but we can
get close enough. Counting crudely,
\[
0 \leq l_f(\T_V) \leq |V|,
\]
where the upper bound of $|V|$ can be obtained by applying Fitch (and observing that, by the earlier assumption, the taxa in the $s$ clade have a different state to the taxa in the $t$ clade). Also,

\[
M \cdot k + M^{*}/2 \leq l_f(\T_E) \leq M \cdot k + |E| + 2 + M^{*}/2.
\]
(The $|E|+2$ on the right-hand side of the above expression is an upper bound on the number of mutations
incurred on the $|E|$ edges leaving $\rho$ and the 2 edges leaving $\rho^{*}$.) From this it follows that
\[
M \cdot k + M^{*}/2 - |V| \leq l_f(\T_E) - l_f(\T_V) \leq M \cdot k + |E| + 2 + M^{*}/2
\]
Assuming $M$ has been chosen such that $M > |V| + |E| + 2$, we observe
that $l_f(\T_E) - l_f(\T_V)$ - and thus also $d_{MP}^{2}(\T_V,\T_E)$ -
will be maximized by selecting $k$ as large as possible i.e. by selecting a maximum-size cut. So,
\[
MAXCUT(G) - \frac{|V|}{M} \leq \frac{ d_{MP}^{2}(\T_V,\T_E) - M^{*}/2 }{M} \leq  MAXCUT(G) + \frac{|E|+2}{M}
\]
Taking $M > 3(|V|+|E|+2)$ is therefore sufficient to yield Equation \ref{eq:round}. This completes the reduction.

\end{proof}

\noindent
In fact, we can show that computing $d_{MP}^{i}(\T_1,\T_2)$ is NP-hard for every fixed integer $i \geq 2$. This will be proven as a corollary of the following theorem, which is the main result of this section.

\begin{theorem}
\label{thm:mainhard}
Computing $d_{MP}(\T_1,\T_2)$ is NP-hard.
\end{theorem}
\begin{proof}
This time we reduce from \textsc{CUBIC MAXIMUM INDEPENDENT SET}. This is the
problem of computing a maximum-size independent set (i.e. set of mutually non-adjacent
vertices) of an undirected 3-regular graph $G=(V,E)$. The problem is NP-hard (and APX-hard), see again \cite{alimonti2000}. Again, the restriction to 3-regular graphs is not essential but simplifies the
proof a little. Let $MIS(G)$ be the size of a maximum-size independent set. It is easy to show
that $MIS(G) \geq |V|/4$ (by adding an arbitrary vertex to be in the independent set,
deleting this vertex and its at most three neighbors, and then iterating this process). Let $V = \{v_1, \ldots, v_n\}$ and $E = \{e_1, \ldots, e_m\}$.
For convenience we assume $|V| \geq 8$, so $MIS(G) \geq 2$, and that $G$ is not bipartite
(because computing $MIS(G)$ on bipartite $G$ can be done in polynomial time).

We construct the trees $\T_V$ and $\T_E$ as described in the previous reduction. We then
change the trees slightly; the changes are shown as dotted lines in Figures \ref{fig:tv} and \ref{fig:te}. To $\T_V$ we
add a new vertex $w$ and we add the edge $\{\rho, w\}$. Then we add $|E|$ taxa directly
beneath $w$. In $\T_E$, for each $e_i$, we subdivide the edge entering $e_i$, and
attach one of the new taxa to the newly created node; we will call this taxon $w_i$. Each tree will have
$3|V|M + |E| + M^{*}$ taxa. 

As in the previous reduction we choose $M^{*}$ sufficiently large such that, for all optimal characters $f$,
$l_f(\T_V) < l_f(\T_E)$.  Using the same refinement argument as before, this can be achieved simply by choosing $M^{*}$ to be any
even number such that:
\[
M^{*}/2 - 1 > 3|V|M + |E|.
\]
Consider an optimal character $f$ for $\T_V$ and $\T_E$. Due to our choice of $M^{*}$ we know $l_f(\T_V) < l_f(\T_E)$.
Hence we can apply Lemma \ref{lem:monochrome} to transform $f$, in polynomial time, such that for each vertex in $\{w,s,t,v_1,\ldots,v_n\}$, all the children of that vertex (in $\T_V$) have the same state. From now
on we will refer to these states as \emph{colors}. Indeed, $f$ naturally induces a coloring of the vertices
of $G$: vertex $v_i$ of $G$ is allocated the color of its children in $\T_V$. In fact, we argue that for appropriately large $M$, $f$ must induce a \emph{proper} coloring of $G$, i.e. a coloring such that every edge is bichromatic. To see why this is, suppose the induced coloring of $G$ has
at least one monochromatic edge. Then a (crude) upper bound on $l_{f}(\T_E)$, and thus also $l_{f}(\T_E) - l_{f}(\T_V)$, is
\[
(|E|-1)M + 3|E| + 2 + M^{*}/2.
\]
(The $3|E| + 2$ term is obtained by assuming that on all edges entering the $e_i$ and $w_i$ nodes,
on all edges leaving $\rho$, and on the two edges leaving $\rho^{*}$, incur mutations.)
 
On the other hand, if the induced coloring is proper, a \emph{lower} bound on $l_{f}(\T_E)$ is
\[
|E|M + M^{*}/2
\]
and an upper bound on $l_{f}(\T_V)$ is $|V|+3$. So if the coloring is proper,
\[
l_{f}(\T_E) - l_{f}(\T_V) \geq |E|M + M^{*}/2 - |V| - 3.
\]
So to enforce that the induced coloring is proper we simply choose $M$ such that
\[
|E|M + M^{*}/2 - |V| - 3 > (|E|-1)M + 3|E| + 2 + M^{*}/2.
\]
A choice of $M > 3|E| + |V| + 5$ is adequate.

Henceforth we can focus our attention on optimal characters $f$ that have the Lemma \ref{lem:monochrome} property
and that induce proper colorings in $G$. Note that, if we apply Fitch to $f$ on $\T_V$, each vertex in $\{w,s,t,v_1,\ldots,v_n\}$ will
be allocated the same color as its children. The colors assigned to $\rho, \rho^{*}$ and $u$ depend on the exact colors
used for $\{w,s,t,v_1,\ldots,v_n\}$. Fitch will assign $\rho$ a color that occurs most frequently
on $\{w,v_1, \ldots, v_n\}$ and (as discussed earlier) will break any ties in the top-down phase of the algorithm. If the most
frequent color on $\{w,v_1, \ldots, v_n\}$ occurs $z$ times, there will be exactly $(n+1)-z$ mutations incurred on the
edges leaving $\rho$.

Let $c$ be the color of $w$ and its children in $\T_V$. By a very careful analysis of Fitch's algorithm on $\T_E$, we can see that Fitch will definitely assign color $c$ to $\rho$ in $\T_E$. Specifically, note that in the
bottom-up phase of Fitch, $c$ will definitely be in the set of labels allocated to the parent of each $w_i$. Color $c$ is thus a most frequently occurring color amongst the children
of $\rho$. There can be no other most frequently occurring color $c'$, because $c'$ would then have to appear as an
endpoint of every edge in the proper coloring, implying that $G$ is bipartite, which we have apriori excluded as a possibility. Hence $c$ is already chosen by Fitch in its bottom-up phase to be the definite label of $\rho$.

Consequently, in its top-down phase the parents of the $w_i$ in $\T_E$ will also be labelled $c$. This means that, for an edge $e_i$ of $G$, there will be a mutation on the edge entering
$e_i$ (in $\T_E$) \emph{if and only if} neither endpoint of $e_i$ is colored $c$. Now, suppose that $c$ is one of the colors used to color $\{v_1,\ldots,v_n\}$ in $\T_V$. Then there exists an edge $e_i$ which has exactly one endpoint colored $c$, meaning that no mutation is incurred on the edge entering $e_i$ in $\T_E$. If we relabel the $|E|$ taxa beneath $w$ to a new color that does not appear elsewhere in $f$, then the parsimony score of $\T_V$ under
$f$ increases by at most one, while the parsimony score of $\T_E$ under $f$ definitely increases by (at least) one. The
latter increase occurs because Fitch will now ascribe mutations to \emph{all} the edges entering the $e_i$ vertices, without causing a reduction in the number of mutations on the edges $\{ \rho^{*}, \rho \}$, $\{\rho^{*}, u\}$.

 Hence, the new $f$ is still optimal and still has the earlier derived properties, plus the new property that $c$, the color of $w$, can be assumed to be distinct from the colors used on $\{v_1, \ldots, v_n\}$.

Now, knowing that Fitch will definitely assign $c$ to $\rho$ in $\T_E$, we can argue without loss of generality that the colors
assigned to $s$ and $t$ in $\T_V$ can also be distinct from $c$. To see this, suppose one of $s$ and $t$ is colored $c$, without loss of generality let
this be $s$. (We note that $s$ and $t$ must have different colors, otherwise $f$ cannot possibly be optimal). Then
Fitch will definitely color $\rho^{*}$ and $u$ with color $c$ in $\T_E$ and there will be no mutation on edge $\{\rho^{*}, \rho\}$ or $\{\rho^{*},u\}$ in
$\T_E$. Recoloring the clade $s$ with a  new color $d$ will cost at most one new mutation in $\T_V$, and at least one
new mutation in $\T_E$ (on the edge $\{\rho^{*}, \rho\}$ or $\{\rho^{*},u\}$), so the new $f$ will still be optimal.

Hence we can assume that in $f$ neither $s$ nor $t$ has color $c$ (where $c$ is the color of $w$), and that $c$ is not a color used on $\{v_1, \ldots, v_n\}$.

Now, consider the most frequently occuring colors on $\{ w, v_1, \ldots, $\\$v_n \}$. Given that $c$ is only used
to color $w$, there must exist at least one most frequently occurring color distinct from $c$. Furthermore, one of these
colors must be used to color $s$ or $t$. To see why this is, let $c' \neq c$ be one of the most frequently occuring
colors; if we relabel (say) clade $s$ with $c'$ we lower the parsimony score of $\T_V$, without lowering the parsimony score
of $\T_E$, contradicting the optimality of the character. 

We have thus reached the point that we can assume that in $f$ neither $s$ nor $t$ has color $c$ (where $c$ is the color of $w$), that $c$ is not a color on $\{v_1, \ldots, v_n\}$, and that either $s$ or $t$ has the same color as some
most frequently occurring color on $\{v_1, \ldots, v_n\}$. All such $f$ have the same parsimony score on $\T_E$. The
only degree of freedom left is to minimize the parsimony score of $\T_V$. This can be achieved by choosing a proper
coloring of $G$ such that the frequency of the most frequently occurring color is maximized. This is essentially
equivalent to constructing a maximum-size independent set of $G$: give all the vertices in the independent set the
same color, and all the other vertices distinct colors. When this is done, $|V| - MIS(G)$ mutations will
be incurred in $\T_V$ on the edges feeding into the $v_i$ nodes, 1 mutation will be incurred on the edge entering
$w$, and in total 1 mutation will be incurred on the two edges leaving $u$. Hence we finally arrive at the following conclusion:
\[
d_{MP}(\T_V,\T_E) = M^{*}/2 + 1 + |E| + |E|M - ((|V|+1) - MIS(G) + 1).
\]

From this the value $MIS(G)$ can easily be computed, and the independent set itself can be obtained by applying Fitch to $\T_V$ (after having applied all the necessary transformations to the character) and returning all $v_i$ that are labelled with the same color as $\rho$. This completes the reduction.
\end{proof}

The above theorem assists us in extending the result given in Lemma \ref{lem:boundedhard}.

\begin{corollary}
\label{cor:fiveishard}
Computing $d_{MP}^{i}(\T_1,\T_2)$ is NP-hard for every fixed integer $i \geq 2$.
\end{corollary}
\begin{proof}
The case $i=2$ is proven in Lemma \ref{lem:boundedhard}. The case $i=3$ requires an ad-hoc extension of that lemma and we defer this to the appendix. Assume then that $i \geq 4$. In the above proof we argued that in $f$ neither $s$ nor $t$ has color $c$ (where $c$ is the color of $w$), that $c$ is not a color on $\{v_1, \ldots, v_n\}$, and that either $s$ or $t$ has the same color as some
most frequently occurring color on $\{v_1, \ldots, v_n\}$. To achieve this, we are free to use some colors
more than once. Now, observe that $G$ has chromatic number 3. This follows from Brook's Theorem \citep{brooks1941colouring}, because $G$ has maximum degree 3, is not equal to $K_4$, and is not bipartite. Moreover, the results in \citep{catlin1979brooks} state that such a graph can be colored with 3 colors such that the most frequently occuring color induces a maximum independent set. Let \emph{blue} be the color corresponding to the vertices of the maximum independent 
set, and let \emph{red} and \emph{green} be the remaining two colors. We use without loss of generality \emph{blue} and \emph{red} to color $s$ and $t$, and introduce a new fourth color \emph{yellow} to color $w$. Hence at most 4 colors are needed to optimize $d_{MP}(\T_V,\T_E)$, even when more are available, from which the result follows.
\end{proof}

\section{Discussion}
\label{sec:discussion}
In this article we have explored several properties of the MP distance measure. We have also proven that this new metric is NP-hard. A hardness result for the case of \emph{binary} trees remains elusive, although we strongly believe that this is also NP-hard and shall elaborate upon this in a forthcoming publication. In any case, the NP-hardness of the metric is not, in itself, a reason to cease investigating it. Recent years have seen an explosion of academic interest in overcoming in practice the theoretical intractability of attractive measures such as SPR distance and hybridization number (see e.g., \citep{whidden2013fixed,van2012practical}). This article should therefore be viewed as the start signal for deeper research into MP distance. What exactly is its relation to SPR and other phylogenetic measures? How far can the NP-hardness be tamed in practice? In future research we will tackle these and other questions.

\section*{Acknowledgement}
We wish to thank David Bryant for bringing the topic to our attention. Also, MF wishes to thank Bhalchandra Thatte and Mike Steel for helpful discussions on the topic, and SK thanks Nela Lekic for useful discussions on the relationship between chromatic number and independence number.

\section{Appendix: Generalizing Lemma \ref{lem:boundedhard} to 3 states}

\begin{lemma}
\label{lem:bounded34hard}
Computing $d_{MP}^{3}(\T_1,\T_2)$ is NP-hard.
\end{lemma}
\begin{proof}
We note that determining whether a $4$-regular
graph $G=(V,E)$ is $3$-colorable, is NP-hard \citep{dailey1980uniqueness}. (The regularity restriction is not strictly essential but again makes the analysis cleaner). Let $G$ be a 4-regular graph. Let $BICHROM(G)$ be the maximum number of
bichromatic edges possible, ranging over all (proper or non-proper) colorings of $V$ with
at most $3$ colors. Observe that $G$ is $3$-colorable if and only if $BICHROM(G)=|E|$.
The following reduction shows that $BICHROM(G)$ can be computed using an oracle
for $d_{MP}^{3}(\T_1,\T_2)$, from which hardness follows. The reduction differs only from that used in
Lemma \ref{lem:boundedhard} in the choice of  $M$ and $M^{*}$.

Clearly, $|E| = 2|V|$. Let $\{0,1, 2\}$ be the set of $3$ character states. During this proof we will write ``any character'' as shorthand for ``any character with at most $3$ states''.

The high-level idea is to construct two trees $\T_1$ and $\T_2$, henceforth
referred to as $\T_V$ and $\T_E$, where $\T_V$ encodes the vertices and $\T_E$ the
edges of $G$. In $\T_V$ the character states $\{0,1,2\}$ will be used to indicate the color assigned to each vertex. The mutations induced in $\T_E$ will
be used to count the number of bichromatic edges induced by the coloring. 

Let $V = \{v_1, \ldots, v_n\}$ and $E = \{e_1, \ldots, e_m\}$. Throughout the reduction we will utilize two large numbers,
$M$ and an even number $M^{*}$, such that $n << M << M^{*}$ but such that
both are still at most $poly(n)$. In due course we will explain
how these numbers are calculated.  Both trees will have
$4|V|M + M^{*}$ taxa.

To construct $\T_V$ we first introduce vertices $\rho^{*}, \rho, s, t, u$ and $\{v_1, \ldots, v_n\}$. We add edges $\{\rho^{*},\rho\}$, $\{\rho^{*},u\}$, $\{u,s\}$, $\{u,t\}$ and
$\{\rho, v_i\}$ for $1 \leq i \leq n$. We connect each $v_i$ to $4M$ taxa. We connect
$s$ to $M^{*}/2$ taxa and $t$ to $M^{*}/2$ taxa.

To construct $\T_E$ we introduce vertices $\rho^{*}, \rho, u$ and $\{e_1, \ldots, e_m\}$.
We add edges $\{\rho^{*},\rho\}$, $\{\rho^{*},u\}$ and
$\{\rho, e_i\}$ for $1 \leq i \leq m$. To $u$ we connect the $M^{*}$ taxa that were
connected to $s$ and $t$ in $\T_V$. Next we introduce $mM$ vertices $e_{i,j}$ for
$i \in \{1,\ldots,m\}$ and $j \in \{1,\ldots,M\}$. We connect each vertex $e_i$ to
all the $e_{i,j}$, for $j \in \{1,\ldots,M\}$. To each $e_{i,j}$ we connect two taxa, representing the endpoints (of the $j$th copy) of edge $e_i$. To determine
which two taxa these are, suppose in $G$ the edge $e_i$ is connected to vertices $v_a$ and
$v_b$. Then one of the two taxa is taken from the clade of taxa we connected to $v_a$ in $\T_V$, and the other from the clade of taxa beneath $v_b$ in $\T_V$. The exact mapping
chosen does not matter.
This completes the construction.

For $x \in \mathbb{R}$ let $round(x)$ be $x$ rounded to the nearest integer. (By
construction we will actually only use values of $x$ that are at most 1/3 above or below the nearest integer; this will become clearer later).We shall prove the following:
\begin{equation}
\label{eq:round2}
BICHROM(G) = round \bigg ( \frac{ d_{MP}^{3}(\T_V,\T_E) - M^{*}/2 }{M} \bigg )
\end{equation}
Moreover, we will show how any optimal character $f$ can
be transformed in polynomial-time into a coloring that has $BICHROM(G)$ edges. 

Note that in $\T_V$ the subtree rooted at $u$ is a refinement of the subtree rooted at $u$ in $\T_E$. This means that, for any character $f$, a (crude) upper bound on $l_f(\T_V) - l_f(\T_E)$
is $4|V|M$. To see this, observe that for any character the number of mutations incurred in $\T_V$ in the
subtree rooted at $u$, is less than or equal to the number of mutations incurred in the corresponding subtree of $\T_E$.
Therefore mutations in this subtree can never contribute to an increase in $l_f(\T_V) - l_f(\T_E)$. 
Hence, an upper bound on $l_f(\T_V) - l_f(\T_E)$ can be
achieved by maximizing the number of mutations that occur in $\T_V$ on edges that are not in this subtree, and minimizing the number of mutations that occur in $\T_E$ on edges that are not in this subtree. A (trivial) lower bound on the latter is 0, while a (trivial) upper bound on the former is $4|V|M$ (e.g. by assigning the same state to all internal nodes of $\T_V$).

On the other hand, consider a character $f$ such that the $M^{*}/2$ taxa
underneath $s$ in $\T_V$ are allocated (wlog) state 0, and all other taxa are allocated (wlog) state 1. Then
$l_f(\T_V)=1$ and $l_f(\T_E) \geq M^{*}/2$.  For this reason we choose $M^{*}$ such that
\[
M^{*}/2 - 1 > 4|V|M.
\]
After choosing $M^{*}$ this way we know that $d_{MP}^{3}(\T_V,\T_E) \geq  (M^{*}/2)- 1$ and, more importantly, that for every optimal character $f$, $l_f(\T_V) < l_f(\T_E)$.

From Observation \ref{obs:boundedmonochrome} we can therefore assume that, for each vertex of $\T_V$ in
the set $\{s,t,v_1, \ldots, v_n\}$, all the taxa beneath the vertex are allocated the same state by $f$. (The same observation tells us that it is safe to assume that
all the taxa in the $s$ clade have the same state, and all the taxa in the
$t$ clade have the same state, and that these two states are distinct.)

Such a character naturally induces a coloring of the vertices. The core observation is that if an edge $e_i$ is bichromatic (i.e. the colors at its endpoints are different) then $f$ will induce at least $M$ mutations in the subtree of
$\T_E$ rooted at $e_i$. On the other hand, an edge $e_i$ is monochromatic this will induce 0 mutations in the subtree rooted at $e_i$.

More formally, suppose $f$ induces $k$ bichromatic edges. The parameter $k$ does not, in itself,
give us enough information to \emph{exactly} determine $l_f(\T_V)$ and $l_f(\T_E)$, but we can
get close enough. Counting crudely,
\[
0 \leq l_f(\T_V) \leq |V|,
\]
where the upper bound of $|V|$ can be obtained by applying Fitch (and observing that, by the earlier assumption, the taxa in the $s$ clade have a different state to the taxa in the $t$ clade). Also,
\[
M \cdot k + M^{*}/2 \leq l_f(\T_E) \leq M \cdot k + |E| + 2 + M^{*}/2.
\]
(The $|E|+2$ on the right-hand side of the above expression is an upper bound on the number of mutations
incurred on the $|E|$ edges leaving $\rho$ and the 2 edges leaving $\rho^{*}$.) From this it follows that
\[
M \cdot k + M^{*}/2 - |V| \leq l_f(\T_E) - l_f(\T_V) \leq M \cdot k + |E| + 2 + M^{*}/2.
\]
Assuming $M$ has been chosen such that $M > |V| + |E| + 2$, we observe
that $l_f(\T_E) - l_f(\T_V)$ -- and thus also $d_{MP}^{3}(\T_V,\T_E)$ --
will be maximized by selecting $k$ as large as possible i.e. by making as many edges
as possible bichromatic. So,
\[
BICHROM(G) - \frac{|V|}{M} \leq \frac{ d_{MP}^{3}(\T_V,\T_E) - M^{*}/2 }{M} \leq  BICHROM(G) + \frac{|E|+2}{M}
\]
Taking $M > 3(|V|+|E|+2)$ is therefore sufficient. This completes the reduction.
\end{proof}

\bibliographystyle{plainnat}     
\bibliography{bibliographyMPdistance}

\end{document}